%% file: sample-acmsmall.tex
\documentclass[acmsmall]{acmart}

\usepackage{graphicx}
\usepackage{amsmath}
\usepackage{booktabs}
\usepackage{amsthm}
\usepackage{multirow}
\usepackage{subcaption}
\usepackage{nicefrac}
\usepackage[linesnumbered, ruled, commentsnumbered]{algorithm2e}

\usepackage[capitalize]{cleveref}
\crefname{section}{Sec.}{Secs.}
\Crefname{section}{Section}{Sections}
\Crefname{table}{Table}{Tables}
\crefname{table}{Tab.}{Tabs.}

\newtheorem{theorem}{Theorem}
\newtheorem{lemma}{Lemma}

\newcommand{\bones}{\texttt{BONES}\xspace}
\newcommand{\bola}{BOLA\xspace}

\AtBeginDocument{%
  \providecommand\BibTeX{{%
    \normalfont B\kern-0.5em{\scshape i\kern-0.25em b}\kern-0.8em\TeX}}}

\setcopyright{acmlicensed}
\acmJournal{POMACS}
\acmYear{2024} \acmVolume{8} \acmNumber{2} \acmArticle{19} \acmMonth{6}\acmDOI{10.1145/3656014}

\begin{document}

\title{BONES: Near-Optimal Neural-Enhanced Video Streaming}

% \author{Anonymous Authors}
% \renewcommand{\shortauthors}{Anonymous Authors}

\author{Lingdong Wang}
\email{lingdongwang@umass.edu}
\orcid{0009-0007-4672-2243}
\affiliation{
  \institution{University of Massachusetts Amherst}
  \city{Amherst}
  \state{Massachusetts}
  \country{USA}
}

\author{Simran Singh}
\email{s.singh.xzy@gmail.com}
\orcid{0000-0002-0867-2528}
\affiliation{
  \institution{New Jersey Institute of Technology}
  \city{Newark}
  \state{New Jersey}
  \country{USA}
}

\author{Jacob Chakareski}
\email{jacobcha@njit.edu}
\orcid{0000-0003-2428-9518}
\affiliation{
  \institution{New Jersey Institute of Technology}
  \city{Newark}
  \state{New Jersey}
  \country{USA}
}

\author{Mohammad Hajiesmaili}
\email{hajiesmaili@cs.umass.edu}
\orcid{0000-0001-9278-2254}
\affiliation{
  \institution{University of Massachusetts Amherst}
  \city{Amherst}
  \state{Massachusetts}
  \country{USA}
}

\author{Ramesh K. Sitaraman}
\email{ramesh@cs.umass.edu}
\orcid{0000-0003-0558-6875}
\affiliation{
  \institution{University of Massachusetts Amherst}
  \city{Amherst}
  \state{Massachusetts}
  \country{USA}
}


\begin{abstract}
Accessing high-quality video content can be challenging due to insufficient and unstable network bandwidth. Recent advances in neural enhancement have shown promising results in improving the quality of degraded videos through deep learning. Neural-Enhanced Streaming (NES) incorporates this new approach into video streaming, allowing users to download low-quality video segments and then enhance them to obtain high-quality content without violating the playback of the video stream. We introduce BONES, an NES control algorithm that jointly manages the network and computational resources to maximize the quality of experience (QoE) of the user. BONES formulates NES as a Lyapunov optimization problem and solves it in an online manner with near-optimal performance, making it the first NES algorithm to provide a theoretical performance guarantee. Comprehensive experimental results indicate that BONES increases QoE by 5\% to 20\% over state-of-the-art algorithms with minimal overhead. Our code is available at https://github.com/UMass-LIDS/bones.
\end{abstract}

% demonstrating its potential to enhance the video streaming experience for users.

\begin{CCSXML}
<ccs2012>
<concept>
<concept_id>10003033.10003068.10003073.10003074</concept_id>
<concept_desc>Networks~Network resources allocation</concept_desc>
<concept_significance>500</concept_significance>
</concept>
<concept>
<concept_id>10002951.10003227.10003251.10003255</concept_id>
<concept_desc>Information systems~Multimedia streaming</concept_desc>
<concept_significance>500</concept_significance>
</concept>
</ccs2012>
\end{CCSXML}

\ccsdesc[500]{Networks~Network resources allocation}
\ccsdesc[500]{Information systems~Multimedia streaming}

\keywords{adaptive bitrate streaming, Lyapunov optimization, neural enhancement, super-resolution}

\received{August 2023}
\received[revised]{January 2024}
\received[accepted]{March 2024}


\maketitle
% \ramesh{CCS and Key words need modification. Or we could comment out for now.}

\input{introduction}

\input{method}

\input{experiment}
\input{result}
\input{prototype}

\input{related}
\input{conclusion}

\begin{acks}
The work has been supported by the National Science Foundation (NSF) under awards 
% Prof. Ramesh
CNS-1763617, CNS-1901137, CNS-2106463, 
% Prof. Mohammad
CNS-2102963, CAREER-2045641, CNS-2106299,
% Prof. Jacob
CCF-2031881, ECCS-2032387, CNS-2040088, CNS-2032033, and CNS-2106150; by the National Institutes of Health (NIH) under award R01EY030470; and by the Panasonic Chair of Sustainability at the New Jersey Institute for Technology.
\end{acks}

\bibliographystyle{ACM-Reference-Format}
\bibliography{sample-base}

\newpage
\appendix
\input{appendix}

\end{document}

%% file: introduction.tex
\section{Introduction}

Video content dominates the Internet, accounting for more than 65\% of its traffic volume \cite{internet}. However, accessing high-quality video content is often hindered by insufficient and unstable network bandwidth between the video server and video player (i.e., client). This challenge of high-quality video streaming is even more significant when delivering higher-resolution or immersive videos, such as 4K/8K videos, 360-degree videos, and volumetric videos.

The traditional approach to maximizing the user's quality of experience (QoE) is to use an adaptive bitrate (ABR) algorithm. An ABR algorithm typically runs within the video player and ensures that the video plays back continuously (i.e., without rebuffering) at the highest possible quality (i.e., bitrate). To achieve this goal, the ABR algorithm downloads lower-quality video segments when the available network bandwidth is low to prevent rebuffering and downloads higher-quality segments when the bandwidth is high. There has been extensive research on developing ABR algorithms for decades, with many known algorithms such as \bola, Dynamic, Festive, MPC, Pensieve, etc~\cite{bola, bola2, buffer, festive, mpc, pensieve}.

\textbf{Neural-Enhanced Streaming (NES).} Traditional video streaming relies on transmitting videos from the server to the client at the highest possible quality while ensuring continuous playback~\cite{ChakareskiC:06,ThomosCF:10, ChakareskiAWTG:04c}. However,
recent advances in machine learning open up a new possibility of transmitting videos from the server to the client at low quality, and then {\em neurally enhancing} the video quality via deep-learning techniques at the client. Some examples of such neural enhancement include super-resolution \cite{mdsr, swinir}, frame interpolation\cite{spacetime, unconstrain}, video inpainting \cite{lama, videoinpaint}, video denoising \cite{videnn, fastdvdnet}, point cloud upsampling \cite{punet, pugan}, and point cloud completion \cite{pfnet}. NES methods incorporate this new approach into video streaming, enabling a tradeoff between communication resources for transmitting high-quality videos and computational resources for neural enhancement. Unlike ABR algorithms that only decide the quality of the video segment to download, an NES control algorithm decides on {\em both} the download quality and the enhancement option, for each video segment. NES is particularly advantageous in improving the worst-case user experience under poor network conditions. Additionally, it allows video providers or clients to save bandwidth by lowering the transmission bitrate and performing enhancement afterward.

\textbf{Prior work on NES.} While not as extensively studied as ABR algorithms, recent works on NES use reinforcement learning (RL) \cite{nas, improve, sr360, sophon, srcache} or heuristic approaches \cite{presr, yuzu}. However, these methods do not have theoretical guarantees for their performance. Further, prior works only consider one enhancement option during inference, usually the one bringing the highest quality gain in real-time. This restricts the possible design space of enhancement options and disregards the broader benefits of utilizing diverse enhancements over a longer time horizon. Finally, existing NES methods are complex to deploy and slower to converge. For example, RL-based methods have high training costs but may still not adapt well to real-world scenarios \cite{situ}. Other approaches that rely on model predictive control (MPC) compute a large rigid decision table or heuristically solve an NP-hard problem, leading to an intractable solution \cite{mpc, presr}.

\textbf{Our NES algorithm.} To rectify the above shortcomings, we propose \textbf{B}uffer-\textbf{O}ccupancy-based \textbf{N}eural-\textbf{E}nhanced \textbf{S}treaming (\bones), a client-side NES algorithm for on-demand video streaming. Specifically, \bones downloads video segments from a video provider's server to a client device and then enhances these segments opportunistically using local computational resources. To ensure efficient scheduling of the available bandwidth and computational resources, \bones operates within a novel parallel-buffer system model and solves a Lyapunov optimization problem online. It has a provable near-optimal performance and exploits all available enhancement methods during inference, resulting in superior performance. Besides, \bones has a simple control algorithm with explainable parameters, making it easier to deploy in production systems. 

\textbf{Our Contributions.} We make the following contributions in our work:
\vspace{-0.1cm}
\begin{enumerate}
\item \textit{Joint optimization of download and enhancement decisions.} \bones generalizes the Lyapunov optimization approach of the ABR algorithm \bola \cite{bola, bola2} to the NES problem of scheduling both bandwidth and computational resources. Our proposed parallel-buffer system model and online control algorithm allow for optimizing the download and enhancement decisions jointly while conducting the actual respective operations asynchronously, which largely contributes to the performance enhancements enabled by our framework.

\item \textit{Theoretical guarantees and simple algorithm design.} \bones is the first NES algorithm with a provable guarantee on its performance. Specifically, \bones achieves QoE that is provably within an additive factor of the offline optimal solution. Besides enhanced performance, \bones has the advantages of linear time complexity and simplicity in deployment.

\item \textit{Simulation environment and prototype system implementation.} We implement an efficient simulation environment for large-scale evaluation of NES algorithms, along with a prototype system for real-world examination. Our code is publicly released to advance research in related fields.

\item \textit{Superior experimental performance with low overhead.} Using extensive experiments, we compare  \bones with existing methods under six enhancement settings and four network trace datasets. Our experimental results demonstrate that \bones increase QoE by $3.56\%$ to $13.20\%$ in our simulation evaluation and $4.66\%$ to $20.43\%$ in our prototype evaluation. In comparison to the default ABR algorithm of the \emph{dash.js} video player \cite{dashjs}, \bones can improve its QoE by $7.33\%$, which is equivalent to increasing the average visual quality by 5.22 in the VMAF score. \bones only incurs minimal costs and offers three trade-off options between performance and overhead. The options range from less QoE improvement with zero overhead to maximum benefit with an additional 310-KB download size and 0.68-second startup latency.

\end{enumerate}

%% file: method.tex
\section{Problem Formulation}
% Modern video streaming works by partitioning the video into {\em segments,}
% where each segment plays for a fixed amount of time (say, 4 seconds). The video player sequentially downloads each segment from a video server and renders the segment on the viewer's device. 
% To reduce the chances of rebuffering (i.e., freezing), each segment is downloaded and stored in a {\em download buffer} ahead of when it needs to be rendered. Each segment is encoded in multiple qualities. And an ABR algorithm chooses the quality to download in an online fashion with the goal of reducing rebuffering events and optimizing the QoE. As noted earlier, an NES algorithm also performs neural enhancements for downloaded segments prior to rendering. In this section, we first propose a system model for our NES system, capturing both the download and enhancement process in \cref{sec:system}. Then, we formulate an optimization problem in \cref{sec:objective} that is solved by our NES algorithm \bones. The main notations used in this paper are summarized in \cref{tab:notation}.
We now propose a system model and pose the optimization solved by \bones. The main notations used in this paper are summarized in \cref{tab:notation}.

\begin{table}[t]
\centering
\caption{Summary of main notations.}
\label{tab:notation}
\begin{tabular}{|c|l|}
\hline
Notation & Description \\
\hline
$t_k$ & the $k$-th time slot\\
$K_n$ & index of the time slot where the $n$-th segment is downloaded \\
$T_k$ & duration of time slot $t_k$\\
$d_i$ & binary indicator to select the $i$-th download bitrate \\
$e_j$ & binary indicator to select the $j$-th enhancement method \\
$Q^d(t_k)$ & download buffer level at time slot $t_k$ \\
$Q^e(t_k)$ & enhancement buffer level at time slot $t_k$ \\
$u^d(i,t_k)$  & download utility of the $i$-th bitrate in time slot $t_k$ \\
$u^e(i, j, t_k)$  & enhancement utility of the $j$-th method for the $i$-th bitrate in time slot $t_k$ \\
$\Tilde{u}^e(i, j, t_k)$ & timely enhancement utility of the $j$-th method for the $i$-th bitrate in time slot $t_k$ \\

$p$ & duration of a video segment \\
$B_i(t_k)$ & the $i$-th download bitrate in time slot $t_k$\\
$S_i(t_k)$ & size of the video segment with the $i$-th bitrate in time slot $t_k$\\
$\omega(t_k)$ & average bandwidth in time slot $t_k$\\
$t^e(i,j)$  & processing time to enhance a segment with the $i$-th bitrate using the $j$-th method \\
$T_{\text{end}}$ & playback finishing time \\
$V$  & parameter to control the trade-off between Lyapunov drift and penalty \\
$Q^d_{\max}$ & maximum download buffer capacity \\
$u_{\max}$ & maximum utility of a video segment \\
$T_{\min}$ & minimum time slot duration \\
$T_{\max}$ & maximum time slot duration \\
$u_n$ & total utility of the $n$-th video segment \\
$\phi_n$ & rebuffering time to download the $n$-th video segment \\
$\gamma$ & hyper-parameter to control the trade-off between utility and smoothness \\
$\beta$ & hyper-parameter to linearly control $V$ \\

\hline
\end{tabular}
\end{table}

\subsection{System Model}
\label{sec:system}

Modern video streaming works by temporally partitioning the video into {\em segments,}
where each segment plays for a fixed amount of time (say, 4 seconds). The video player sequentially downloads each segment from a video server and renders the segment on the viewer's device. 
To reduce the chance of rebuffering (i.e., freezing), each segment is downloaded and stored in a {\em download buffer} ahead of when it needs to be rendered. Each segment is encoded in multiple qualities. And an ABR algorithm chooses the quality to download in an online fashion with the goal of reducing rebuffering events and optimizing the QoE. As noted earlier, an NES algorithm also performs neural enhancements for downloaded segments prior to their rendering.

We show the system model of a traditional ABR algorithm within the video player and contrast that with the system model of \bones in \cref{fig:system}. Relative to a traditional ABR system that solely schedules bandwidth resources with one download buffer, our NES system also incorporates an extra buffer and control flow to manage computational resources. In the \bones system model, the \emph{download buffer} stores video segments that have been downloaded, and the \emph{enhancement buffer} stores enhancement tasks waiting for computational resources. We first introduce the download process in the lower branch and then the enhancement process in the upper branch.

\begin{figure}[t]
    \centering
    \includegraphics[width=\linewidth]{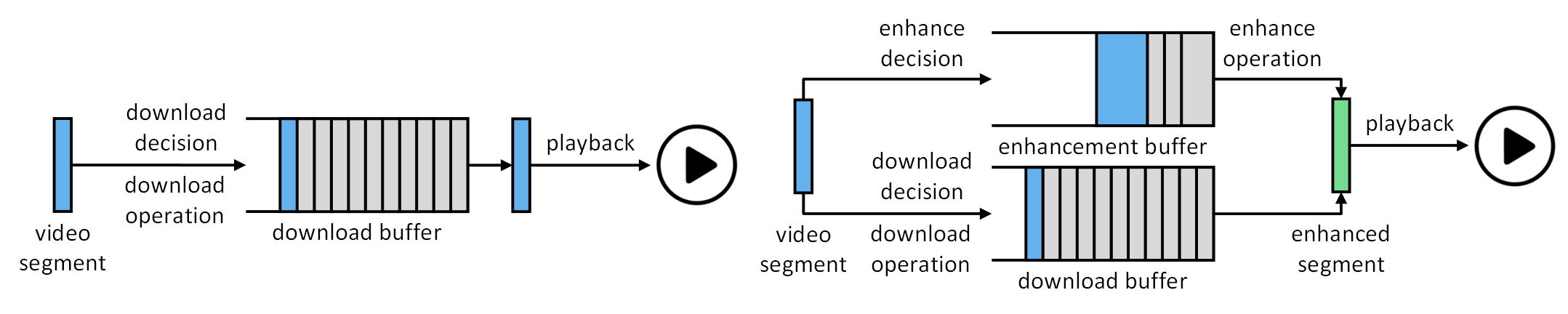}
    \caption{System models of the traditional ABR algorithm (left) and \bones (right). }
    \label{fig:system}
\end{figure}

Suppose the video segment is indexed by $n \in [N]$, and the time slot is indexed by $k \in [K_N]$, where $K_N$ represents the index of the slot where the $N$-th (last) segment of the video content is downloaded. At the beginning of each time slot $t_k$, \bones makes both download and enhancement decisions for the next video segment.

The download decision determines whether to download the next segment in this time slot and at which bitrate to download it. Formally, we represent the possible choices for this decision variable using a vector $\mathbf{d} = (d_1, \cdots, d_I)$, where the indicator $d_i \in \{0, 1\}$, $\sum_{i=1}^I d_i \leq 1$, where $I$ denotes the number of possible download bitrate options. In particular, $\sum_i d_i = 0$ indicates no download, while $d_{i} = 1$ indicates that the next segment will be downloaded at the $i$-th available bitrate $B_i(t_k)$ at time slot $t_k$. A high download bitrate improves the video quality but consumes more network bandwidth, leading to a longer download time.

If the decision is to download, the segment will be retrieved at the desired bitrate and then added to the end of the download buffer, as illustrated in \cref{fig:system}. The next time slot will commence immediately after the completion of the download. If there is no download to be performed in the present time slot, \bones waits for $ \Delta $ milliseconds and starts the next slot. The video content is played back at a constant rate from the download buffer. However, if the download buffer is empty, the playback will freeze until new content is available, which results in rebuffering and negatively impacts the quality of experience (QoE) of the user.

Let  $p$ be the time duration for a video segment to be played back, expressed in milliseconds. Assume that the $i$-th download option has been selected, i.e., ${d_i=1}$. Then, the segment's size in bits can be expressed as ${S_{i}(t_k) = B_i(t_k) p}$. Next, let $\omega(t_k)$ denote the average network bandwidth during time slot $t_k$ in Kbps. We can then compute the duration of time slot $t_k$, i.e., the download time of the segment as $T_k = S_i(t_k) / \omega(t_k)$ milliseconds. Note that $T_k$ also captures the length of downloaded content consumed by the client due to video playback during time slot $t_k$. Finally, let $Q^d(t_k)$ denote the download buffer level, i.e., the total length of segments remaining in the buffer. We can formulate the download buffer dynamic as follows: 
\begin{equation}
\label{eq:Qd}
    Q^d(t_{k+1}) = \max \left[Q^d(t_k) - T_k, 0 \right] + \sum_i d_i p.
\end{equation}
At the start of time slot$t_k$, we also decide whether to enhance the video segment and which enhancement method to apply. The enhancement decision is defined as a vector $\mathbf{e} = (e_1, \cdots, e_J)$, where the indicator $e_j \in \{0, 1\}$. But slightly different from the download decision, we have $\sum_{j=1}^J e_j = 1$, where $J$ denotes the number of possible enhancement methods. We set the first enhancement method ($e_1$) as "no enhancement" with zero finishing time and zero effect and $e_{j} = 1$ indicates that the $j$-th available enhancement method will be applied to the video segment. An enhancement method using a larger deep-learning model typically offers greater quality improvement but takes more time to complete due to higher computational complexity.

A video segment will be pushed into the enhancement buffer at the same time as it is pushed into the download buffer. When doing so, we are actually registering a new computation task to enhance this segment with the chosen enhancement method, and letting this task wait for resources in the enhancement buffer. The length of the video segment in the enhancement buffer is set to be the expected time required to finish its enhancement, denoted by $t^e(i, j)$. This computation time $t^e(i, j)$ is a function of the download option $i$ and the enhancement option $j$, as both the downloaded video quality and enhancement method can affect the computation speed. The estimated computation time is derived by profiling the enhancement method on the client device. Since the playback duration of a segment is usually different from the processing time of its enhancement task, a video segment often has different lengths in the two buffers, as illustrated in \cref{fig:system}.

It is important to note that, though the enhancement decision is made before a segment enters the enhancement buffer, the actual computation for the enhancement takes place when the video segment leaves the buffer. This enables \bones to optimize download and enhancement decisions jointly but perform the actual operations asynchronously in a non-blocking manner, which leads to superior performance. This approach also allows us to use the Lyapunov optimization framework \cite{renewal} which can synchronously control multiple queues in a renewal process.

The departure of a segment from the enhancement buffer indicates the completion of its enhancement task. The enhancement buffer has the same departure rate as the download buffer, because in 1 second, we will playback 1-second video content and complete a 1-second computational task. Every video segment must go through both download and enhancement buffers to become an enhanced segment before playback. But the enhanced segment could be the same as the original if "no enhancement" is selected.
Let the enhancement buffer level $Q^e (t_k)$ capture the aggregate length of computational tasks waiting in the buffer measured in milliseconds. Now, we can formulate the temporal evolution of the enhancement buffer as:
\begin{equation}
\label{eq:Qe}
    Q^e(t_{k+1}) = \max \left[Q^e(t_k) - T_k, 0 \right] + \sum_{ij} d_i e_j t^e(i, j).
\end{equation}

\subsection{Optimization Objective}
\label{sec:objective}

There are three main goals in video streaming - improve video quality, reduce the rebuffering ratio, and avoid buffer overflow. Similarly to \bola \cite{bola, bola2}, we incorporate these three goals into a Lyapunov optimization problem. 
To improve the video quality, we aim to maximize a time-average utility function defined as
\begin{equation}
\label{eq:utility}
\begin{split}
     \Bar{u} & \triangleq \lim_{N \to \infty} \frac{\mathbb{E} \left\{ \sum_{kij} d_i u^d(i, t_k) + d_i e_j \Tilde{u}^e(i, j, t_k) \right\}}{\mathbb{E} \left\{ T_{\text{end}} \right\}} \\
    &= \frac{ \lim_{K_N \to \infty} \frac{1}{K_N} \mathbb{E} \left\{ \sum_{kij} d_i u^d(i, t_k) + d_i e_j \Tilde{u}^e(i, j, t_k) \right\}}{ \lim_{K_N \to \infty} \frac{1}{K_N} \mathbb{E} \left\{ \sum_{k=1}^{K_N} T_k \right\}},
\end{split}
\end{equation}
where $u^d(i, t_k)$ is the base utility of a video segment, $\Tilde{u}^e(i, j, t_k)$ is the extra utility obtained by neural enhancement, and $T_{\text{end}}$ is the playback finishing time. While the proposed algorithm works for any size of the video sequences, in our theoretical analysis, we further assume there are infinite video segments, i.e., $N \to \infty$. 
Because the gap between the total playback time and total download time is the time to drain all the content out of the download buffer, we have $\mathbb{E} \left\{ T_{\text{end}} \right\} - \mathbb{E} \left\{ \sum_{k=1}^{K_N} T_k \right\} \leq Q^d_{\max}$, where $Q^d_{\max}$ is the maximum download buffer capacity. Since $Q^d_{\max}$ is finite, we further have $\lim_{K_N \to \infty} \frac{\mathbb{E} \{ T_{\text{end}} \}}{\mathbb{E} \left\{ \sum_{k=1}^{K_N} T_k \right\}} = 1$. Based on this equation and the theory of renewal processes \cite{discrete}, we derive the last formula in \cref{eq:utility}. The final utility function is the expected sum of the download and enhancement utility in each time slot, averaged by the expected time slot duration.

Different from existing myopic NES methods that only consider the present time slot, we allow for selecting non-real-time enhancements and buffering of computation tasks over a longer time horizon. However, we must make sure that the enhancement of a segment is finished before its playback. 
Utilizing the property that the two buffers have the same departure rate, we know that an enhancement requiring $t^e(i, j)$ to finish will miss the playback deadline if $Q^e(t_k) + \sum_{ij} d_i e_j t^e(i, j) > Q^d (t_k)$. So we abandon such options by assigning them a negative infinity utility. The final enhancement utility function is expressed in \cref{eq:ue}, where $u^e(i, j)$ represents an enhancement method's utility improvement pre-measured by the video provider and transmitted to the client via metadata. Specially, we have $t^e(i, 1) = 0, u^e(i, 1, t_k) = 0, ~\forall i, k$ for the "no enhancement" option.
\begin{equation}
\label{eq:ue}
\Tilde{u}^e(i, j, t_k) =
\left\{
    \begin{aligned}
    &- \infty, ~Q^e(t_k) + \sum_{ij} d_i e_j t^e(i, j) > Q^d (t_k), \\
    &u^e(i, j),  \text{~otherwise}.
    \end{aligned}
    \right.
\end{equation}

In order to reduce the rebuffering ratio, we aim to maximize the time-average playback smoothness function defined in \cref{eq:smooth}. In particular, by maximizing the ratio of the total video length $\sum_{kij} d_i e_j p$ and the total playback time $T_{\text{end}}$, we are minimizing their difference, i.e., the total rebuffering time.

\begin{equation}
\label{eq:smooth}
\begin{split}
    \Bar{s} &\triangleq \lim_{N \to \infty} \frac{\mathbb{E} \left \{\sum_{kij} d_i e_j p \right\}}{\mathbb{E} \{ T_{\text{end}} \}} \\
    &= \frac{ \lim_{K_N \to \infty} \frac{1}{K_N} \mathbb{E} \left\{ \sum_{kij} d_i e_j p \right\}}{ \lim_{K_N \to \infty} \frac{1}{K_N} \mathbb{E} \left\{ \sum_{k=1}^{K_N} T_k \right\}}.
\end{split}
\end{equation}

To prevent buffer overflow, we require both buffers to be rate stable, which is a relaxation of the strict buffer constraint. Rate stability ensures that the expected input rate is not greater than the expected output rate. We then establish the download buffer constraint as 
\begin{equation}
\label{eq:rated}
    \lim_{K_N \to \infty} \frac{1}{K_N} \mathbb{E} \left\{ \sum_{ki} d_i p \right\} \leq \lim_{K_N \to \infty} \frac{1}{K_N} \mathbb{E} \left\{ \sum_{k} T_k \right\},
\end{equation}
and the enhancement buffer constraint as
\begin{equation}
\label{eq:ratee}
    \lim_{K_N \to \infty} \frac{1}{K_N} \mathbb{E} \left\{ \sum_{k} d_i e_j t^e(i, j) \right\} \leq \lim_{K_N \to \infty} \frac{1}{K_N} \mathbb{E} \left\{ \sum_{k} T_k \right\}.
\end{equation}

Finally, we  formulate the neural enhancement problem as 
\begin{equation}
    \max_{\mathbf{d}, \mathbf{e}} ~~\Bar{u} + \gamma \Bar{s},
    \quad 
    \text{s.t.,} \quad \text{Constraints}~\eqref{eq:rated}, ~\eqref{eq:ratee},
\label{eq:opt}
\end{equation}
where the goal is to maximize the streaming session's utility and smoothness under buffer constraints. The hyper-parameter $\gamma$ controls the trade-off between the smoothness and utility objectives. Note that our optimization target can be reduced to that of \bola by setting $e_1 = 1$ (always choosing "no enhancement"). This implies that our method can function as a neural enhancement streaming algorithm when enhancement is available or a conventional ABR algorithm, otherwise.

\section{\bones: Control Algorithm and Theoretical Analysis}
% In this section, we first present our proposed control algorithm \bones in \cref{sec:control}. Next, we visualize the decision plane of \bones and analyze the impact of system parameters in \cref{sec:plane}. At last, we rigorously prove the performance bound of \bones in \cref{sec:proof}. 
% Finally in \cref{sec:practical}, we introduce heuristics that can further improve \bones in practice.

\subsection{Control Algorithm}
\label{sec:control}

\begin{algorithm}[t]
\DontPrintSemicolon
\KwIn{computation time matrix $T^e \in \mathbb{R}_+^{I \times J}$, total utility matrix $U \in \mathbb{R}^{N \times I \times J}$, segment size matrix $S \in \mathbb{R}_+^{N \times I}$
, maximum download buffer capacity $Q^d_{\max}$, maximum utility $u_{\max}$, segment duration $p$, hyper-parameters $\gamma$ and $\beta$}
Initialize $O \in \mathbb{R}^{I \times J}, \quad V = \beta \nicefrac{((Q^d_{\max} - p) p)}{(u_{\max} + \gamma p)}$ \\
\For{$n \in [1, N]$}{
    $t \leftarrow$ \text{current timestamp} \\
    \For {$i \in [1, I]$}{
        \For {$j \in [1, J]$}{
            $O[i, j] \leftarrow (Q^d(t) p + Q^e(t) T^e[i, j] - V (U[n, i, j] + \gamma p)) / S[n, i]$ \\
            \If{$Q^e(t) + T^e[i, j] > Q^d(t)$}{
                $O[i, j] \leftarrow -\infty$ \\
            }
        }
    }
    $i', j' \leftarrow \arg\min_{i,j} O$ \\
    download the $n$-th segment with the $i'$-th bitrate, wait until the download finishes at $t'$\\
    \If{$Q^e(t') + T^e[i', j'] > Q^d(t')$}{
        $j' \leftarrow 0$ \\
    }
    push the $n$-th segment into the enhancement buffer to enhance it with the $j'$-th method  \\
    push the $n$-th segment into the download buffer \\
    $\Delta t_{\text{sleep}} \leftarrow \max[Q^d(t') + p - Q^d_{\max}, 0] $ \\
    wait for time $\Delta t_{\text{sleep}}$ \\
 }
 \caption{\bones: A joint control algorithm for download and enhancement decisions}
 \label{alg:bones}
\end{algorithm}

 We develop an efficient online algorithm called \bones to compute the download and enhancement decisions. Solving  problem~\eqref{eq:opt} optimally in an online fashion is not practically feasible since the future values of the network bandwidth are uncertain. However, we show that \bones is within an additive factor of the offline optimal solution to the problem~\eqref{eq:opt}.
Inspired by the Lyapunov optimization framework for renewal frames \cite{renewal, book} and \bola, \bones greedily minimizes the time-average drift-plus-penalty for each time slot. Specifically, the objective function involves the Lyapunov drift, defined as
\begin{equation}
\label{eq:od}
     O_D(t_k) = Q^d(t_k) \sum_i d_i p + Q^e(t_k) \sum_{ij} d_i e_j t^e(i, j),\\
\end{equation}
and a penalty function defined as
\begin{equation}
\label{eq:op}
     O_P(t_k) = -\left( \sum_{ij} d_i u^d(i, t_k) + d_i e_j \Tilde{u}^e(i, j, t_k) + \gamma d_i p \right).
\end{equation}
We also need to divide our objective by the time slot duration $T_k=S_i(t_k) / \omega(t_k)$, but the random variable $\omega(t_k)$ can be omitted in the optimization. 

% todo: refine

Altogether, the optimization problem that \bones aims to solve in each time slot is given below:
\begin{align}
\label{eq:alg}
    \min_{\mathbf{d}, \mathbf{e}} &~ \frac{O_D(i, t_k) + V O_P(i, j, t_k)}{\sum_i d_i S_i(t_k)} \\
    \text{s.t.,} &~ d_i \in \{0, 1\} ~\forall i, ~\sum d_i \leq 1,  \nonumber \\
    &~ e_j \in \{0, 1\} ~\forall j, ~\sum e_j = 1, \nonumber
\end{align}
where $V$ is a trade-off factor between the Lyapunov drift and the penalty. Note that \bones relies solely on the download and enhancement buffer levels $Q^d(t_k)$ and $Q^e(t_k)$ in its operation without using any information about the available network bandwidth. Thus, \bones is a ``buffer-occupancy-based'' method.

We present the basic control algorithm of \bones in \cref{alg:bones}. The algorithm receives a matrix of computation time $T^e \in \mathbb{R}_+^{I \times J}$ measured locally, a matrix of total utility $U \in \mathbb{R}^{N \times I \times J}$ from metadata, a matrix of segment size $S \in \mathbb{R}_+^{N \times I}$ from metadata, and some constant numbers and hyper-parameters as indicated in the input argument of \cref{alg:bones}. The entries of the total utility matrix include the aggregate values of the download utility $u^d_n$ and the enhancement utility $u^e_n$ for each segment and enhancement option, i.e., $U[n, i, j] = u^d_n(i) + u^e_n(i, j)$. Note that \bones can also run with average utility and segment size if per-segment data is unavailable.

Concretely, the \bones algorithm operates as follows. The algorithm firstly computes $V$ according to \cref{eq:valv} (Line 2 in \cref{alg:bones}). Then for each video segment, \bones solves \cref{eq:alg} by traversing all possible combinations of download and enhancement options and choosing the one with the minimum objective score (Lines 3 - 13). \bones downloads the video segment at the bitrate given by the solution and waits for the completion of the download. Once the segment is downloaded, \bones will check whether the previously-decided enhancement option is still applicable and pushes a computation task to the enhancement buffer only if it can be finished on time (Lines 15 - 18). \bones also pushes the video segment into the download buffer. Further, \bones will pause downloading until the download buffer can hold one more segment (Lines 20 - 21). Note that the stopping criterion here is implemented differently from the theory. This is because we can know exactly when \bones should stop downloading and how long it should sleep, by assigning $V$ a special upper bound. The impact of different values of $V$ on the performance of \bones will be discussed in \cref{sec:proof}.

\bones has a time complexity of $O(IJ)$ in each iteration, where $I$ and $J$ are the numbers of possible choices for bitrate and enhancement options. Besides, the computation of objective scores can be fully parallelized. As a result, \bones runs efficiently in real-time and is easy to deploy. Based on the primary control algorithm here, we propose two additional heuristics in \cref{sec:ablation} to further improve the practical performance of \bones.

\subsection{Decision Plane of \bones}
\label{sec:plane}
To better understand our algorithm, we depict its 2D decision plane with respect to buffer levels under a hypothetical scenario involving two types of low-resolution (LR) segments and high-resolution (HR) segments. And there is only one type of enhancement method for LR segments. We present four variants of \bones decision planes in \cref{fig:plane} under different parameter settings. 

\textit{Interpreting the decision plane of \bones.} Starting from the lower left corner of a decision plane, the download buffer is empty and the viewer is in urgent need of content. In this case, \bones selects to quickly download LR segments without initiating enhancements. As the download buffer level increases, \bones still downloads LR segments but has time to enhance them and achieve higher visual quality. Once the download buffer level is high enough, \bones can pursue the highest quality by downloading HR segments since enhancement can never be perfect (enabling less quality than HR). When the download buffer is almost full, \bones suspends further downloading to prevent buffer overflow. According to \cref{eq:ue}, \bones will stop enhancing segments if the enhancement buffer level is close to the download buffer level, so it will never enter the upper left white triangle. If enhancements are unavailable on the device, the orange region vanishes and the control plane becomes a 1D function of the download buffer level, which implies \bones is reduced to \bola. To conclude, this simple scenario demonstrates that \bones consistently exhibits reasonable behavior.

\textit{The impact of system parameters.}
\cref{fig:plane} also illustrates the impact of system settings and hyper-parameter settings on the decisions of \bones. Compared with the ``Basic'' setting, we increase the computation speed by  $2\times$ in ``Faster Computation'' as if \bones runs on more powerful computational hardware. In this case, the area of downloading and enhancing LR segments increases, meaning that neural enhancement becomes more desirable with more computational resources available. 

\begin{figure}[!t]
    \centering
    \includegraphics[width=\linewidth]{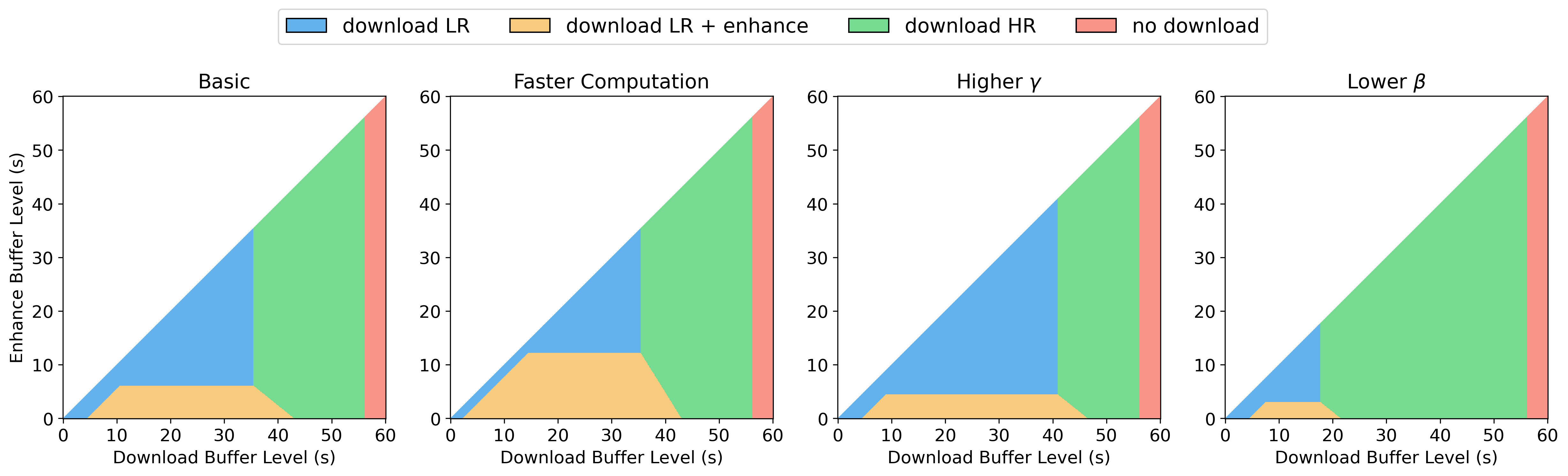}
    \caption{The decision plane of \bones. Parameter settings from left to right are as follows. ``Basic'': $\gamma p = 10, \beta = 1$, $1\times$ computation speed. ``Faster Computation'': $2\times$ computation speed. ``Higher $\gamma$'': $\gamma p = 50$. ``Lower $\beta$'': $\beta = 0.5$.}
    \label{fig:plane}
\end{figure}

Recall that \bones has 2 hyper-parameters --- $\gamma$ and $\beta$. Parameter $\gamma$ trades off between visual quality and playback smoothness (rebuffering ratio) in \cref{eq:opt}. In practice, we tune $\gamma$ together with the constant segment duration $p$. Parameter $\beta$ linearly controls $V$, the trade-off factor between Lyapunov drift and penalty in \bones optimization objective \cref{eq:alg}. In the ``Higher $\gamma$'' setting, we increase $\gamma p$ from 10 to 50, making \bones prefers a lower rebuffering ratio than higher visual quality. As a result, \bones increases the area of downloading LR segments and requests download bitrate in a more conservative manner. 
In the ``Lower $\beta$'' setting, we decrease $\beta$ from 1 to 0.5, suppressing the penalty term in the optimization objective \cref{eq:alg}. Due to the diminishing return effect of visual quality, each bit in an LR video segment provides more visual quality scores than that of an HR segment. Therefore, the penalty term in \cref{eq:alg} prefers downloading LR segments and performing enhancements (see the structure of utility divided by segment size). On the opposite, the Lyapunov drift term prefers HR segments. This explains why \bones will download more HR segments as a consequence of lower $\beta$ and $V$ in this setting. We also provide empirical evidence of how \bones is affected by its hyper-parameters in \cref{sec:hyper}.

\subsection{Performance Bound}
\label{sec:proof}
In this subsection, we rigorously analyze the performance of \bones and establish near-optimal guarantees. The analysis shares the same high-level logical flow with \bola, but since \bones introduces the addition of an enhancement buffer, there are important differences in the details of the analysis. \cref{the:buff} bounds the size of the download buffer and shows that \bones does not violate the buffer capacity. Secondly, \cref{the:perf} bounds the performance of \bones with respect to that of the offline optimal solution.

\begin{theorem}
\label{the:buff}
Assume $Q^d(0)=0, Q^e(0)=0$, and
 $0 < V \leq \frac{(Q^d_{\max} - p) p}{u_{\max} + \gamma p}$
 , where $u_{\max}$ denotes the maximum utility. 
Then, the following holds:
$Q^d(t_k) \leq V \frac{u_{\max} + \gamma p}{p} + p$, 
and $Q^d(t_k) \leq Q^d_{\max}$.
\end{theorem}

We present a proof in \cref{app:proof1}. \cref{the:buff} shows that the download buffer level has a finite bound of $O(V)$ and will never exceed the maximum capacity. The intuition here is to design a special upper bound for $V$, which ensures \bones to select "no download" if the download buffer level is higher than $Q^d_{\max} - p$. In this way,  \bones will download only when the download buffer can accept one more segment. Since the enhancement buffer is a virtual queue without any storage space, its occupancy is not a concern.

In order to prove the performance bound of \bones, we first show that there exists an optimal stationary i.i.d. algorithm independent of the buffer occupancy for problem \cref{eq:opt}, achieving the objective $\Bar{u}^* + \gamma \Bar{s}^*$. Based on that, we can derive the following theorem.

\begin{theorem}
\label{the:perf}
Assume $t^e \leq t^e_{\max}$, $T_{\min} \leq T_k \leq T_{\max}, ~\forall k$, and $t^e_{\max}, T_{\min}, T_{\max}$ are finite. Then, we have:
\begin{equation}
\label{eq:perf}
\Bar{u}' + \gamma \Bar{s}' \geq  \Bar{u}^* + \gamma \Bar{s}^* - \frac{p^2 + (t^e_{\max})^2 + 2 T_{\min} T_{\max} }{2 V T_{\min}},
\end{equation}
where $\Bar{u}' + \gamma \Bar{s}'$ is the objective score of \emph{\texttt{BONES}}.
\end{theorem}

Proofs for the existence of the offline optimal and \cref{the:perf} are given in \cref{app:proof2}. While \bones can achieve an objective $\Bar{u}' + \gamma \Bar{s}'$ for \cref{eq:opt},  the optimal algorithm can achieve $\Bar{u}^* + \gamma \Bar{s}^*$. \cref{the:perf} shows that the performance gap of \bones toward this optimal algorithm is finitely bounded by $O(1/V)$. In other words, \bones can achieve near-optimal performance.

Results in the above theorems show that the performance of \bones depends on the value of parameter $V$. In the implementation of \bones, we use a hyper-parameter $\beta \in [0, 1]$ to linearly control $V$ as follows.
\begin{equation}
\label{eq:valv}
    V = \beta \frac{(Q^d_{\max} - p) p}{u_{\max} + \gamma p}.
\end{equation}
Putting together the results in Theorems~\ref{the:buff} and~\ref{the:perf}, one can observe that \bones has a $O(V, 1/V)$ trade-off between the download buffer level and the performance gap toward the optimal algorithm. It means one can increase $V$ to improve the performance of \bones. Empirical results in \cref{sec:hyper} verify this theoretical observation, showing that the QoE of \bones increases with $V$.  However, increasing $V$ will increase the buffer level, thus increasing the time delay between downloading a segment and playing it back. Live streaming can be negatively affected by such behavior, but it may not be as much of an issue in on-demand video streaming. Besides, there is an upper cap for $V$ to prevent buffer overflow.

%% file: experiment.tex
\section{Experimental Evaluation}
% In this section, we introduce our experiment design. We first define the QoE metrics in \cref{sec:experiment1}, then present the implementation details in \cref{sec:experiment2}. We describe the neural enhancement methods used in our experiments in \cref{sec:setting}. Last, the comparison algorithms are listed in \cref{sec:benchmark}.

\subsection{Performance Metric}
\label{sec:experiment1}

To evaluate the overall performance of video streaming algorithms, we report the experimental results using a commonly-used notion of QoE that includes visual quality, quality oscillation, and rebuffering ratio. More formally, we have
\begin{equation}
\label{eq:qoe}
    \text{QoE}  =  \frac{1}{N} \sum_{n=1}^N u_n  - \alpha_1 \frac{1}{N-1} \sum_{n=1}^{N-1} | u_{n+1} - u_n |  - \alpha_2 \frac{1}{N} \sum_{n=1}^N \phi_n,
\end{equation}
where $u_n$ is the total utility of a segment summing up the download utility $u^d_n$ and the enhancement utility $u^e_n$. And $\phi_n$ is the rebuffering time for each segment. The first term in \cref{eq:qoe} represents the average visual quality, the second term represents the average quality oscillation, and the third term is the average rebuffering time per segment. This notion of QoE is widely used in prior work, e.g., MPC-based methods \cite{mpc, presr, situ, xatu} and RL-based methods \cite{pensieve, nas}.

Although ABR-only algorithms typically assume utility (visual quality) as a function of bitrate, it is not a suitable metric in NES systems. This is because bitrate only applies to compressed videos, not raw pixels after computational processing. Therefore, we choose the VMAF score \cite{vmaf} as the visual quality metric in our experiments, which is closer to human vision than other objective metrics like PSNR or SSIM. VMAF score ranges from 0 to 100, the higher, the better. We assign the highest-resolution video segments with a maximum score of 100 and invalid enhancement options with negative infinity scores. We measure the rebuffering time in milliseconds. And we set the trade-off factors $\alpha_1 = 1, \alpha_2 = 0.1$ as in method \cite{presr}, meaning that 1 QoE score is equivalent to the average visual quality of 1 VMAF score, average quality oscillation of 1 VMAF score, or per-segment rebuffering time of 10 milliseconds ($0.25\%$ rebuffering ratio in our case).

\subsection{Implementation Details}
\label{sec:experiment2}
We develop a unified simulation environment for both ABR and NES algorithms to efficiently examine their performance. Our simulator extends Sabre \cite{bola2}, which was used to evaluate \bola. Our simulation environment and all its algorithms are implemented using Python. And all deep-learning methods are implemented by PyTorch.
We use a 636-second 30-fps video ``Big Buck Bunny'' for streaming. The video is chunked into 4-second segments and encoded in 5 resolutions of 240p/360p/480p/720p/1080p with bitrate of 400/800/1200/2400/4800 Kbps.

\begin{table*}[h]
\centering
\caption{Details of network trace datasets. }
\label{tab:network}  
\begin{tabular}{|c|c|c|c|c|}
\hline
Dataset & 3G & 4G & FCC-SD & FCC-HD \\
\hline
Mean Bandwidth (Kbps) & 1184 & 31431 & 6081 & 17127  \\
Standard Variance of Bandwidth (Kbps) & 818 & 14058 & 11615 & 4018  \\
Number of Traces & 86 & 40 & 1000 & 1000  \\
\hline
\end{tabular}
\end{table*}

We evaluate algorithms using four network trace datasets from Sabre, including 3G traces \cite{3g}, 4G traces \cite{4g}, and two subsets of FCC traces \cite{fcc}. These datasets are widely used, and different studies have adopted different subsets \cite{nas, presr, mpc}. We exclude those 3G traces with an average bandwidth lower than 400 Kbps, the lowest option on our bitrate ladder. We present numerical details of our testing datasets in \cref{tab:network}.

\subsection{Enhancement Settings}
\label{sec:setting}

\begin{table}[ht]
\centering
\caption{Enhancement model details. From left to right: \#layers, \#channels, upscale factor, model size (KB). }
\label{tab:enhance}
\begin{tabular}{|c|c|c|c|c|c|}
\hline
Model & Quality &  240p & 360p & 480p & 720p \\
\hline
\multirow{4}[0]{*}{NAS-MDSR} & low &20,9,4,43 & 20,8,3,36 & 20,4,2,12 & 6,2,1,2 \\
 & medium & 20,21,4,203 & 20,18,3,157 & 20,9,2,37 & 6,7,1,5 \\
 & high & 20,32,4,461 & 20,9,3,395 & 20,18,2,128 & 6,16,1,17 \\
 & ultra & 20,48,4,1026 & 20,42,3,819 & 20,26,2,259 & 6,26,1,41 \\
\hline
\multirow{4}[0]{*}{IMDN} & low & 3,6,4,34 & 3,6,3,29 & 3,6,2,26 & - \\
 & medium & 5,12,4,111 & 5,12,3,103 & 5,12,2,96 & - \\
 & high & 6,32,4,760 & 6,32,3,736 & 6,32,2,719 & - \\
 & ultra & 6,64,4,2824 & 6,64,3,2777 & 6,64,2,2743 & - \\
\hline
\end{tabular}
\end{table}

We utilize five enhancement settings (six, if ``no enhancement'' is counted) to demonstrate that \bones can manage any enhancement method with any amount of computational power. 
We first adopt the pervasively-used enhancement method NAS-MDSR \cite{nas, nemo, stream360}, a content-aware Super-Resolution (SR) model that can overfit one individual video and upscale any resolution to 1080p. We further assume that the client-side computational hardware is an Nvidia GTX 1080ti GPU card. Under this setting, there are five enhancement quality levels (no enhancement, low, medium, high, and ultra) for each of the four low-resolution download options (240p, 360p, 480p, and 720p). More details about the deep-learning model can be found in \cref{tab:enhance}.  Besides, the visual quality and computation speed of each enhancement option are illustrated in \cref{fig:heatmap}, where zero computation speed implies the enhancement is not applicable. Generally, as the download bitrate and the enhancement quality level increase, the visual quality will increase and the computation speed will decrease. However, our algorithm does not rely on this being true.

\begin{figure}[ht]
    \centering
    \begin{subfigure}{0.48\textwidth}
        \centering
        \includegraphics[width=0.9\linewidth]{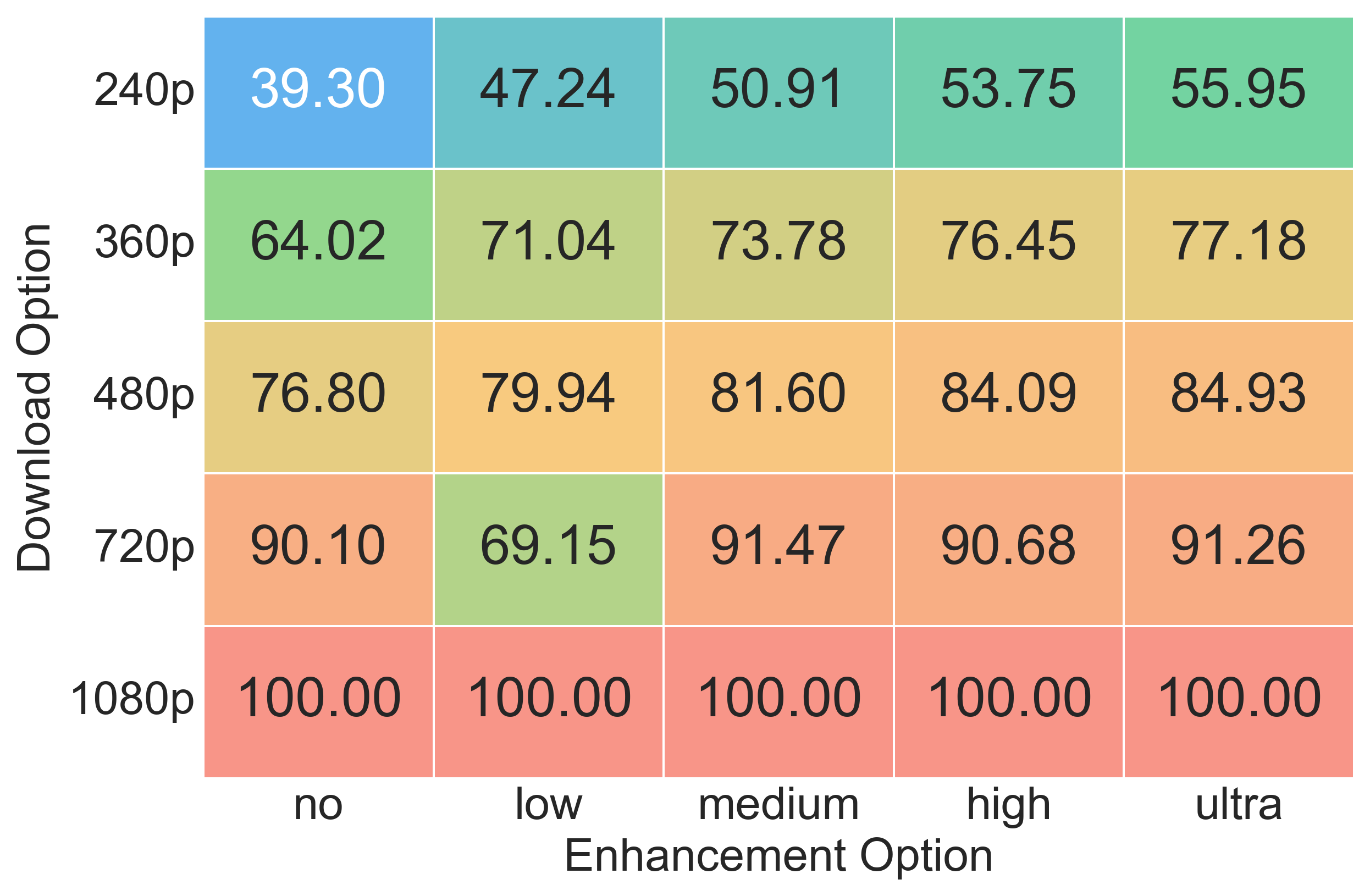}
        \caption{Visual quality (VMAF).}
    \end{subfigure}
    \begin{subfigure}{0.48\textwidth}
        \centering
        \includegraphics[width=0.9\linewidth]{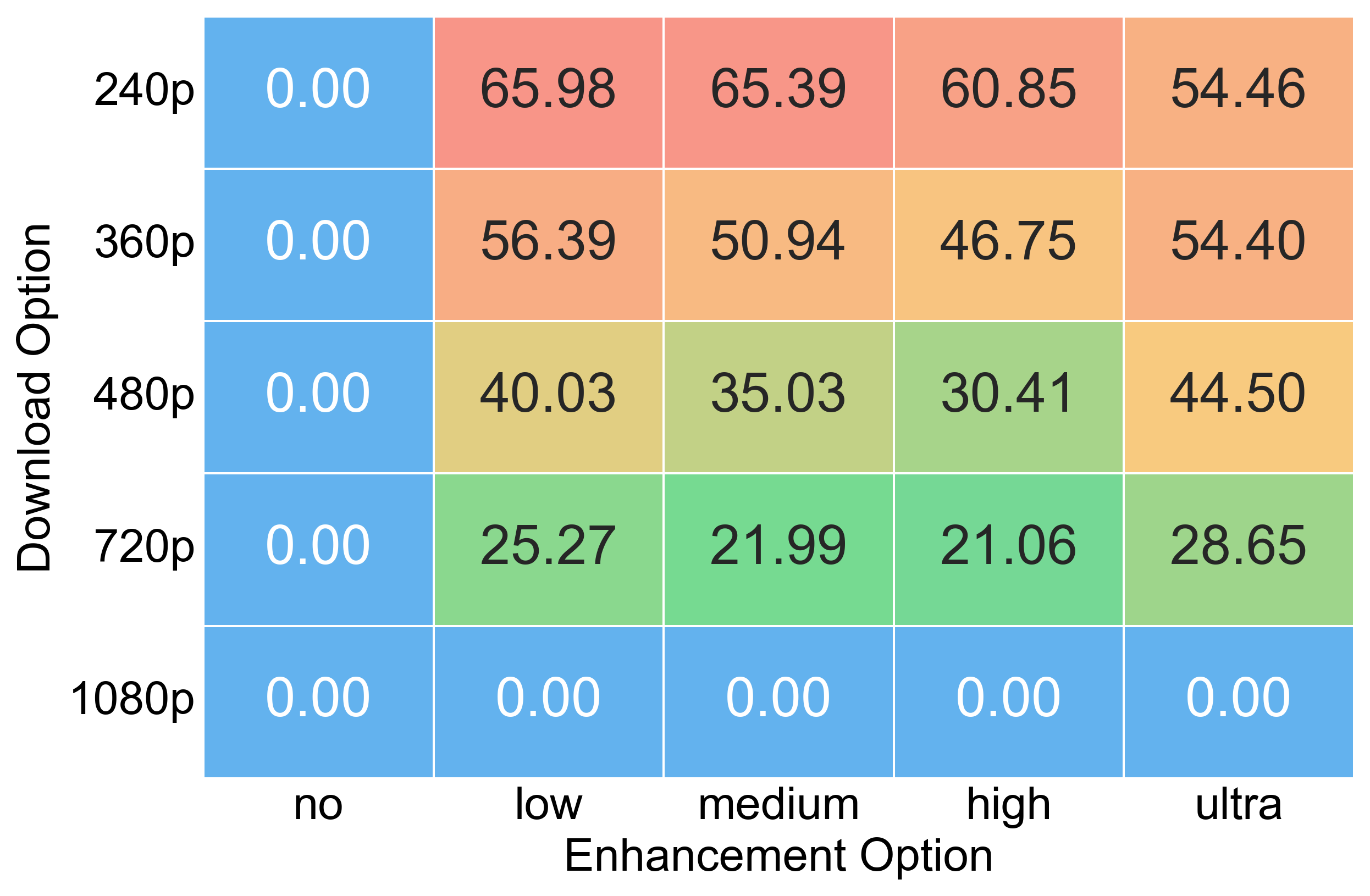}
        \caption{Computation speed (fps).}
    \end{subfigure}
    \caption{Enhancement performance under NAS-MDSR setting.}
    % \vspace{-4mm}
    \label{fig:heatmap}
\end{figure}

Beyond NAS-MDSR, we adopt IMDN \cite{imdn} as another SR enhancement method. We offer 4 quality levels for 240p, 360p, and 480p videos to upsample them to 1080p. The low and medium quality levels are based on the IMDN-RTC structure, while the high and ultra levels are based on the vanilla IMDN. Model details of IMDN can also be found in \cref{tab:enhance}. 
As for the training technique, we explore both content-agnostic IMDN and content-aware IMDN. A content-agnostic IMDN is trained over a generic video dataset DIV2K \cite{div2k}. Such a generic model can enhance any video but bring less quality improvement. In comparison, content-aware IMDN is trained over our target video to overfit its content. This strategy offers better enhancement performance but introduces overhead in training cost and startup latency. A more detailed discussion exists in \cref{sec:overhead}.
The computational hardware for IMDN is either GTX 2080ti or GTX 3060ti GPU card. Altogether, we create 4 enhancement settings for IMDN by traversing all combinations of training techniques and computational hardware.

\subsection{Comparison Algorithms}
\label{sec:benchmark}
We compare \bones with the following ABR-only algorithms:
\begin{enumerate}
    \item \bola \cite{bola} makes download decisions by solving a Lyapunov optimization problem only related to the buffer level. 
    \item Dynamic \cite{bola2} switches between \bola and a throughput-based algorithm based on carefully-designed heuristic rules.
    \item FastMPC \cite{mpc} formulates bitrate adaptation as an MPC optimization problem and makes download decisions according to a pre-computed solution table.
    \item Buffer-based method \cite{buffer} chooses bitrate using a piecewise linear function of buffer level.
    \item Throughput-based method \cite{festive} chooses the maximum bitrate under the estimated network bandwidth.
    \item Pensieve \cite{pensieve} uses deep RL to make download decisions. We train the RL agent on 10,000 randomly selected FCC traces \cite{fcc} using the Asynchronous Advantage Actor-Critic (A3C) algorithm with identical training settings as in the original paper.
\end{enumerate}

We further augment ABR-only algorithms with a greedy enhancement strategy, denoted by the symbol $*$ on the right of the algorithm names. Specifically, we let the ABR algorithms make download decisions while simultaneously choosing the enhancement option that brings the most utility improvement in real-time. To avoid interrupting playback, enhancements will only be applied if it can be finished on time. This greedy enhancement strategy is simple and will not affect the download performance of ABR algorithms. However, it only leads to sub-optimal performance because the download decision maker and the enhancement decision maker share no knowledge with each other. 
Beyond ABR algorithms, we also compare \bones with the following NES algorithms:

\begin{enumerate}
    \item NAS \cite{nas} controls both download and enhancement processes using an RL agent. The NAS RL agent is designed and trained similarly to Pensieve's, with the additional consideration of a greedy real-time enhancement strategy. We did not reproduce the scalable model download approach in NAS, which can be viewed as providing more enhancement options.
    \item PreSR \cite{presr} solves an MPC problem online with the heuristic to pre-fetch and enhance only ``complex'' segments of the video.
    
\end{enumerate}

%% file: result.tex
\section{Simulation Results}

\subsection{Overall Performance}
\label{sec:result1}

\begin{table*}[ht]
\centering

\caption{Overall performance comparison under different enhancement settings. Evaluated metrics include visual quality (VMAF), quality oscillation (VMAF), rebuffering rate (\%), and composite QoE. * denotes the ABR method is augmented with greedy enhancement. \bones achieves the highest QoE under all applicable settings. The optimal results are marked in \textbf{bold}. }
\label{tab:result}

\begin{tabular}{|c|c|c|c|c|c|c|c|c|c|}
\hline

Method & Qual.  & Osc.  & Rebuf. & QoE & Method & Qual.  & Osc. & Rebuf.  & QoE \\
\hline

\multicolumn{5}{|c|}{No Enhancement} & \multicolumn{5}{|c|}{NAS-MDSR, Content-Aware, GTX 1080ti} \\
\hline

\bola  & 82.90 & 4.05 & 2.21 & 69.98 & \bola* & 85.45 & 3.66 & 2.21 & 72.93 \\
Dynamic & 84.62 & 3.73 & 2.43 & \textbf{71.14} & Dynamic* & 86.62 & 3.43 & 2.43 & 73.44 \\
FastMPC & 87.90 & \textbf{3.33} & 4.22 & 67.65 & FastMPC* &  89.70 & \textbf{3.22}  & 4.22 & 69.56 \\
Buffer &  85.27 & 4.10  & 2.93 & 69.44 & Buffer* & 87.47 & 3.66 & 2.93 & 72.07  \\
Throughput &  80.10 & 3.59 & \textbf{1.93} & 68.77 & Throughput* & 83.43 & 3.33 & \textbf{1.93} & 72.36 \\
Pensieve & \textbf{89.23} & 4.01 & 4.17 & 68.50 & Pensieve* & \textbf{91.16} & 3.47 & 4.17 & 70.96 \\
NAS  & - & - & - & -  & NAS  & 89.20 & 4.40 & 2.92 & 73.08 \\
PreSR & - & - & - & -   & PreSR & 89.24 & 4.17 & 4.65 & 66.45  \\
\bones & - & - & - & -    & \bones  & 89.23 & 3.07 & 2.52 & \textbf{76.05} \\

\hline
\multicolumn{5}{|c|}{IMDN, Content-Agnostic, GTX 2080ti} & \multicolumn{5}{|c|}{IMDN, Content-Aware, GTX 2080ti} \\
\hline
\bola* & 83.99 & 3.92 & 2.21 & 71.20 & \bola* & 84.84 & 3.76 & 2.21 & 72.22 \\
Dynamic* & 85.62 & 3.63 & 2.43 & 72.25 & Dynamic* & 86.11 & 3.49 & 2.43 & 72.87 \\
FastMPC* & 88.46 & \textbf{3.27} & 4.22 & 68.26 & FastMPC* & 88.88 & \textbf{3.33} & 4.22 & 68.63 \\
Buffer* & 86.09 & 3.97 & 2.93 & 70.39 & Buffer* & 86.83 & 3.78 & 2.93 & 71.32  \\
Throughput  & 81.76 & 3.45 & \textbf{1.93} & 70.57  & Throughput* & 82.59 & 3.35 & \textbf{1.93} & 71.50 \\
Pensieve* & \textbf{89.90} & 3.79 & 4.17 & 69.39  & Pensieve* & \textbf{90.19} & 3.71 & 4.17 & 69.76 \\
NAS   & 88.18 & 4.66 & 2.92 & 71.81 & NAS & 88.35 & 4.57 & 2.92 & 72.06 \\
PreSR  & 88.65 & 4.39 & 4.57 & 65.95 & PreSR  & 88.79 & 4.20 & 4.51 & 66.53\\
\bones & 88.40 & 3.52 & 2.80 & \textbf{73.64} & \bones & 88.97 & 3.47 & 2.55 & \textbf{75.29} \\

\hline
 \multicolumn{5}{|c|}{IMDN, Content-Agnostic, GTX 3060ti} & \multicolumn{5}{|c|}{IMDN, Content-Aware, GTX 3060ti} \\
\hline
\bola* & 84.85 & 3.77 & 2.21 & 72.21  & \bola* & 84.75 & 3.83 & 2.21 & 72.05 \\
Dynamic* & 86.39 & 3.49 & 2.43 & 73.14 & Dynamic* & 86.49 & 3.42 & 2.43 & 73.32 \\
FastMPC* & 89.12 & \textbf{3.15} & 4.22 & 69.06 & FastMPC* & 89.17 & \textbf{3.21} & 4.22 & 69.03 \\
Buffer* & 86.91 & 3.80 & 2.93 & 71.38 & Buffer* & 86.77 & 3.85 & 2.93 & 71.19  \\
Throughput* & 82.58 & 3.36 & \textbf{1.93} & 71.48 & Throughput* & 83.46 & 3.34 & \textbf{1.93} & 72.38 \\
Pensieve* & \textbf{90.50} & 3.60 & 4.17 & 70.18 & Pensieve* & \textbf{91.05} & 3.50 & 4.17 & 70.83 \\
NAS   & 88.98 & 4.39 & 2.92 & 72.88 & NAS & 89.17 & 4.38 & 2.92 & 73.08 \\
PreSR  & 89.18 & 4.11 & 4.63 & 66.53 & PreSR  & 88.85 & 4.20 & 4.52 & 66.55 \\
\bones  & 88.49 & 3.37 & 2.64 & \textbf{74.52}   & \bones  & 89.38 & 3.12 & 2.47 & \textbf{76.36} \\

\hline
\end{tabular}
\end{table*}

\begin{figure}[t]
  \centering
    \includegraphics[width=\textwidth]{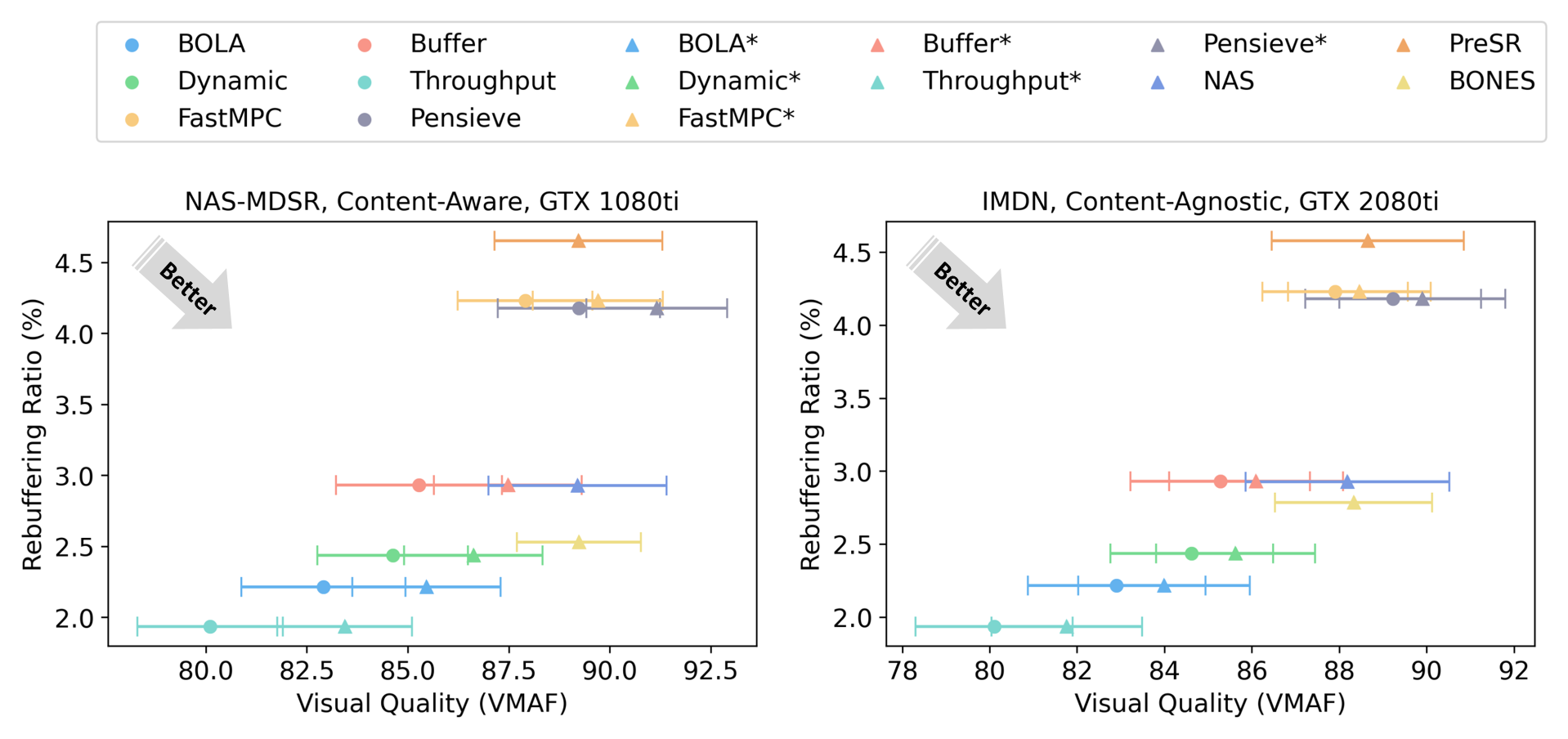}
   \caption{Performance comparison under two enhancement settings. The error bar represents the average visual quality oscillation. 
   Higher quality, lower oscillation, and lower rebuffering ratio are better.
   * denotes the ABR method is augmented with a greedy enhancement strategy.}
   \label{fig:nasall}
\end{figure}

We present the simulation results of \bones and other benchmark methods averaged across all network trace datasets in \cref{tab:result} and \cref{fig:nasall}. In \cref{tab:result}, we report the QoE score defined in \cref{eq:qoe} as well as its three components (visual quality, quality oscillation, rebuffering ratio).
The results are grouped under six settings, one setting without enhancement, and five with different enhancement settings, one with NAS-MDSR, and four with IMDN as discussed in \cref{sec:setting}. We provided visual representations of the three QoE components for each method under two enhancement settings earlier in \cref{fig:nasall}.
Among ABR-only algorithms, we find that Dynamic, the default ABR algorithm of \emph{dash.js} video player ~\cite{dashjs}, reaches the best balance between the three metrics and achieves the highest QoE. Yet \bones can still improve its QoE by up to $7.33\%$, which is equivalent to increasing the average visual quality by 5.22 VMAF score or decreasing the rebuffering ratio by $1.30 \%$.

By augmenting ABR-only algorithms with the greedy enhancement strategy, their QoE benefits from the additional usage of computational resources. The greedy enhancement strategy will not affect the download behavior of an ABR algorithm, thus keeping its rebuffering ratio unchanged. But by upgrading the low-quality downloaded segments, neural enhancement increases the visual quality and reduces quality oscillation, leading to an average of $2.81\%$ QoE improvement across all kinds of augmentations.  However, this enhancement strategy is decoupled from the download decision and thus leads to sub-optimal performance. In contrast, \bones makes joint decisions and outperforms all the augmented ABR algorithms.

As for NES methods, we find that both NAS and PreSR perform well under high-bandwidth scenarios but poorly if the network condition is weak. In~\cref{sec:result2}, we further scrutinize the robustness of different algorithms under different network trace datasets. Our experimental results suggest that it's hard for NAS to adapt to data distribution far from the training set. Similarly, \cite{situ} reports that RL-based methods have limited generalization ability and may not adapt to heavy-tailed real-world network traces. In comparison, \bones is a control-theoretic algorithm that does not require any learning process, making it easy and robust to deploy under versatile scenarios.

We note that the unsatisfactory performance issue of PreSR comes from its reliance on sub-optimal heuristics. PreSR formulates an NP-hard MPC optimization problem and solves it heuristically online. As a result, PreSR can perform even worse than its backbone method FastMPC, as it only explores limited solutions in the decision space. In contrast, our method \bones is guaranteed to outperform its backbone \bola, which is theoretically shown in \cref{sec:control} and empirically verified here.

By studying the four variants of IMDN enhancement settings, we find that the content-aware model performs better than the content-agnostic one as it overfits the content of a specific video. Nevertheless, content-aware enhancement requires costly training toward an individual video and increases the startup latency due to model downloading before video streaming. We postpone the detailed overhead analysis to \cref{sec:overhead}.  Besides, we find that the performance of neural enhancement increases with better computational hardware (GTX 3060ti than GTX 2080ti), which is intuitive since more computational resources enable more powerful enhancements.

In conclusion, our method \bones consistently outperforms existing ABR and NES methods by $3.56\%$ to $13.20\%$ in QoE averaged across all network conditions and all enhancement settings.

\subsection{Sensitivity to Network Condition}
\label{sec:result2}

\begin{figure}[t]
    \centering
    \includegraphics[width=\textwidth]{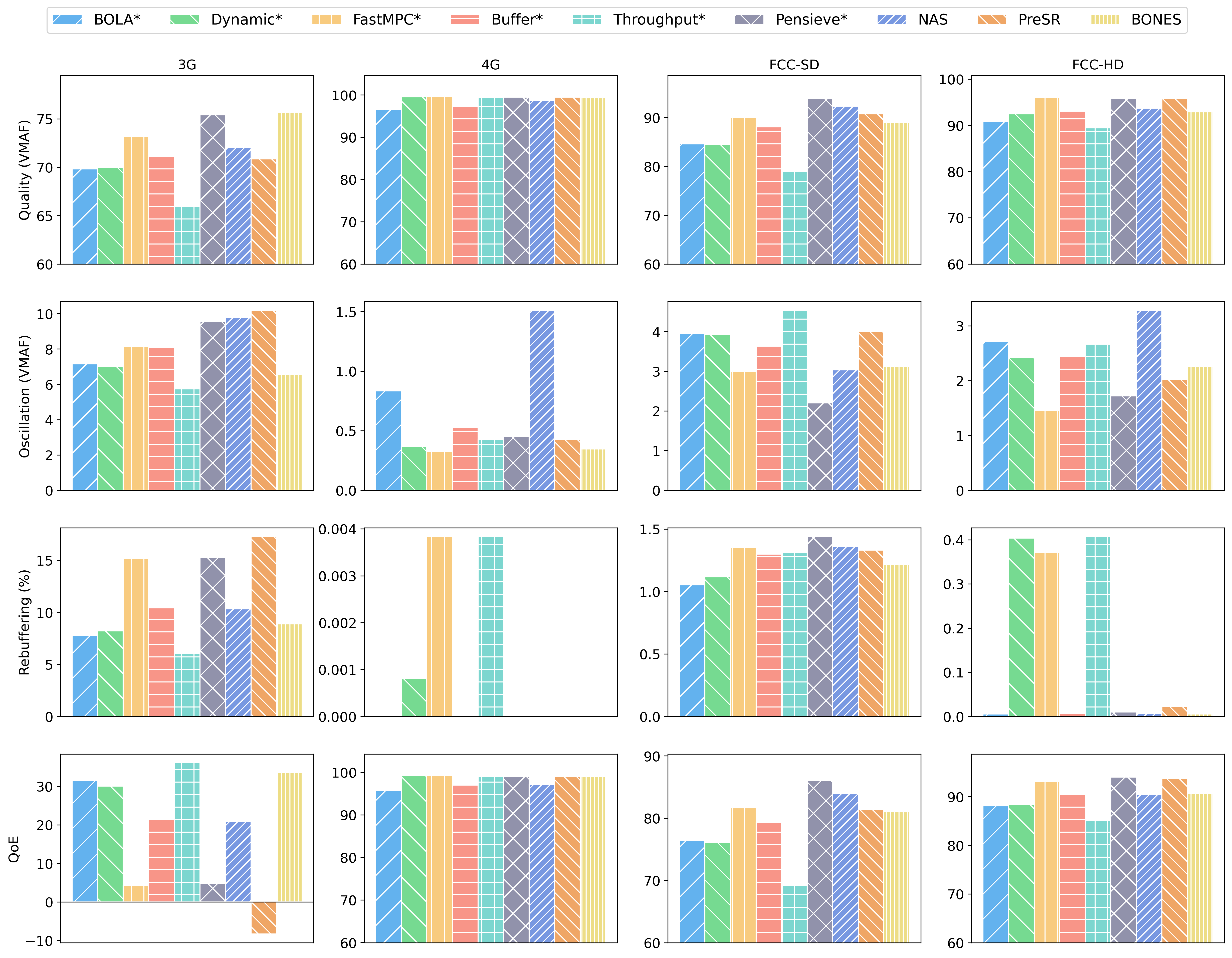}
    \label{fig:dataset}
   \caption{Performance comparison on different network trace datasets. \bones consistently deliver high QoE across diverse network conditions, including challenging ones.}
\end{figure}

To investigate the sensitivity of different algorithms to the network conditions, we further break down the results under the NAS-MDSR enhancement setting in~\cref{tab:result}  across four different network traces and report them in~\cref{fig:dataset}. Based on the results, we find that the performance of video streaming algorithms varies drastically with network conditions. For example, the throughput-based algorithm outperforms others on the low-bandwidth 3G dataset due to its conservative request of bitrate. However, this behavior does not generalize well to high-bandwidth conditions. Especially the throughput-based method performs worst on the FCC-SD dataset because the high variance of bandwidth leads to a misprediction of throughput. On the other hand, FastMPC, Pensieve, and PreSR perform well under good network conditions via aggressive bitrate requests. But they incur severe rebuffering events on 3G datasets. Unlike other methods, \bones behaves robustly across all network conditions, maintaining a stable QoE even under challenging low-bandwidth scenarios.

\subsection{Practical Improvement}
\label{sec:ablation}

In this subsection, we introduce two additional heuristics to further advance the performance of \bones and report their empirical improvement.

\textit{Download Process Monitoring. } Since the network bandwidth can vary abruptly in practice, a decision may become out-of-date when the corresponding video segment is under download. 
For example, downloading a high-quality segment may incur rebuffering if the bandwidth drops in midway. In such a case, it is wise to abort the current download and select a lower-quality segment. We implement this in a similar way to \bola, allowing our algorithm to monitor the bandwidth variations during download processes and regret its previous decision if necessary. During the download process of a segment, \bones periodically recomputes the current decision's objective score via a modified version of \cref{eq:alg}, where $S_i(t_k)$ is replaced with the remaining size to download. \bones also computes the objective scores of other options. If a better option is found, \bones will abandon the ongoing download and switch to the new solution. 

\textit{Automatic Hyper-Parameter Tuning. } We observe that the hyper-parameter setting of \bones is sensitive to the network condition. Under poor network conditions, it is advisable to adopt a conservative strategy with lower quality and a lower rebuffering ratio, while under good conditions, an aggressive strategy is preferred. Inspired by Oboe \cite{oboe}, we develop an automatic hyper-parameter tuning mechanism for \bones in accordance with different network conditions. In the offline phase, we generate synthetic network traces with different average bandwidths (from 200 to 5000 Kbps), different bandwidth variances (500 to 5000 Kbps), and different latency (20 to 100 ms). Then we search for the optimal hyper-parameter $\beta$ (from 0.1 to 1) and $\gamma p$ (from $10$ to $100$) under each network condition that maximizes the QoE. During the online deployment phase, \bones estimates the network condition using the exponentially weighted moving average before each decision interval. \bones then adjusts its hyper-parameters to the offline optimal accordingly.

\textit{Empirical Results.}
We conduct an ablation study for practical improvement techniques and present the results under the NAS-MDSR enhancement setting in \cref{tab:ablation}. Note that those versions without automatic hyper-parameter tuning are assigned with the optimal hyper-parameters found by grid search ($\beta=1$, $\gamma p = 10$). As experimental results imply, Download Process Monitoring can reduce the rebuffering ratio by regretting previous decisions during bandwidth valley. While Automatic Hyper-Parameter Tuning can contribute to the performance of \bones mostly in visual quality. 
When both techniques are applied, \bones could achieve the maximum QoE improvement of $1.48\%$. In summary, this ablation study verifies the effectiveness of our practical improvement techniques.

\begin{table}[t]
\centering

\caption{Ablation study for practical improvements. }
\label{tab:ablation}

\begin{tabular}{|c|c|c|c|c|}
\hline
Method & Qual. (VMAF) & Osc. (VMAF) & Rebuf. (\%) & QoE \\
\hline
Basic & 88.64 & 3.17 & 2.62 & 74.95  \\
Monitor & 88.76 & \textbf{3.03} & \textbf{2.52} & 75.62 \\
AutoTune & 89.22 & 3.16 & 2.58 & 75.71 \\
Monitor+AutoTune & \textbf{89.23} & 3.07 & \textbf{2.52} & \textbf{76.05} \\

\hline
\end{tabular}

\end{table}

\subsection{Sensitivity to Hyper-Parameters }
\label{sec:hyper}

\begin{figure}[ht]
    \centering
    \includegraphics[width=\linewidth]{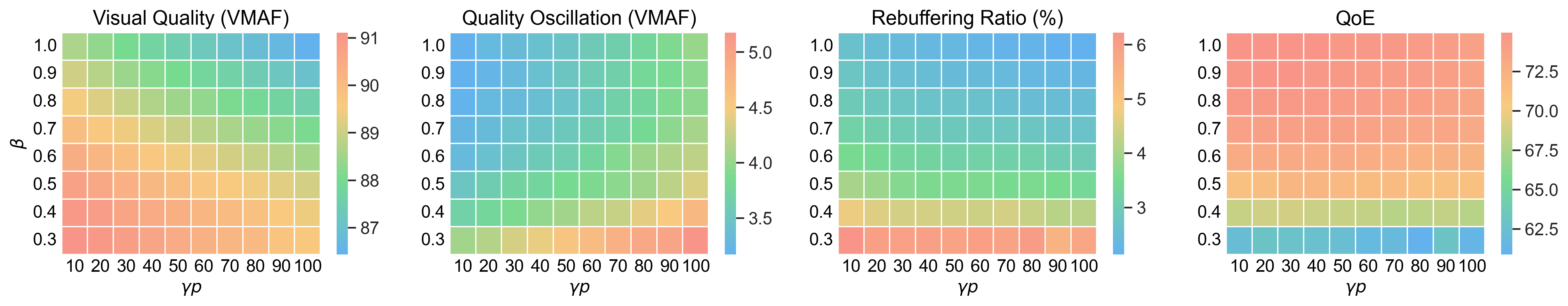}
    \caption{Performance of \bones with respect to hyper-parameter settings.}
    \label{fig:hyper}
\end{figure}

In this section, we comprehensively investigate the effect of hyper-parameters. It is worth noting that, in practice, one may leverage the Automatic Hyper-Parameter Tuning scheme in \cref{sec:ablation} instead of manually adjusting parameters.
From the experimental results under NAS-MDSR enhancement settings in \cref{fig:hyper}, we find that increasing $\gamma$ generally leads to lower visual quality, higher quality oscillation, and a lower rebuffering ratio. As explained in \cref{sec:plane}, a higher $\gamma$ makes \bones download more low-bitrate segments, trading quality for playback smoothness. Besides, due to the diminishing return of visual quality, low-bitrate segments span wider in quality scores. Thus downloading more low-quality segments brings more oscillations.
We also find that a higher $\beta$ leads to lower visual quality, lower quality oscillation, and a lower rebuffering ratio. From our theoretical analysis in~\cref{sec:plane} we know that, by increasing the weight on the penalty term of optimization objective \cref{eq:alg}, \bones downloads more low-bitrate segments and performs more enhancements, resulting in lower quality and fewer rebuffering events. Additionally, since neural enhancement can improve visual quality and reduce the visual quality disparities between segments, increasing $\beta$ also alleviates quality oscillations. Lastly, empirical results verify that higher $\beta$ (so does $V$) positively contributes to the QoE. This conclusion roughly aligns with the observation in \cref{the:perf}, despite the QoE here being defined in a slightly different way than the optimization objective.

\subsection{Overhead Analysis}
\label{sec:overhead}

In this section, we comprehensively analyze the overhead of \bones, including offline costs of training enhancement models, online costs of downloading them during playback, and other costs.

\begin{figure}[t]
  \centering
    \includegraphics[width=0.9\linewidth]{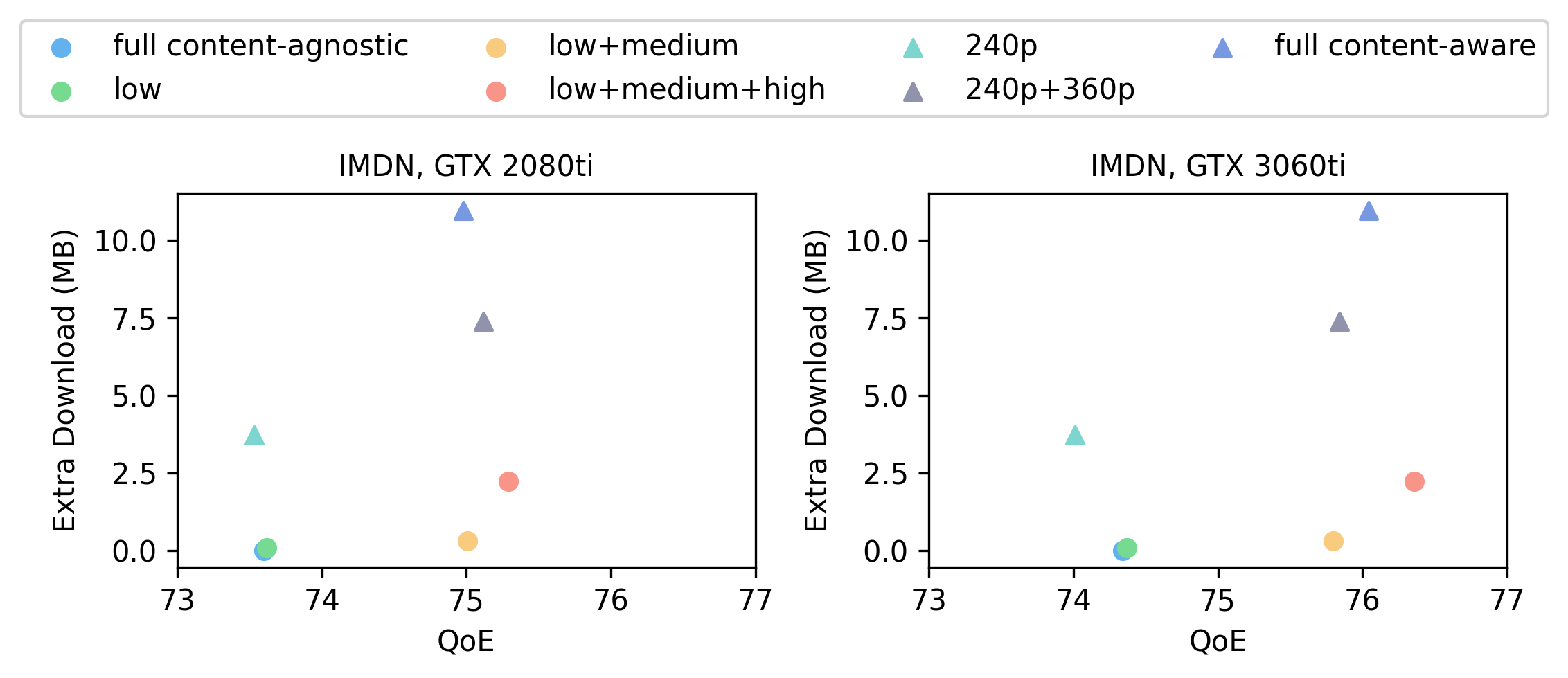}
   \caption{Performance of \bones with respect to the additional download size.} 
   \label{fig:overhead}
\end{figure}

\textit{Offline Costs.} 
A content-agnostic model is trained over a generic video dataset and enhances all video content. So it only incurs a one-time training cost without the need for fine-tuning.
In contrast, a content-aware model needs to be fine-tuned for a particular video using a content-agnostic backbone.
NAS \cite{nas} reports that a typical fine-tuning will take 10 minutes for each video. The enhancement benefit will cover this additional cost after streaming the video for 30 hours. As a rule of thumb, one may apply content-aware enhancements only for popular videos.

\textit{Online Costs.}
Online costs of \bones include the additional bandwidth consumption and startup latency in downloading extra models.
The content-agnostic enhancement has all models pre-installed on the client device.
So this cost is only applicable to content-aware enhancement as it requires specific models per video.
Unlike other methods, \bones allow choosing different enhancement methods from a pool of enhancement methods on the fly. This behavior leads to more QoE improvement but requires pre-downloading more models, which may further increase the online cost. Therefore, we study the trade-off between performance and extra download size in \cref{fig:overhead}.

We plot QoE of content-agnostic IMDN and the additional download overhead (0 MB) in \cref{fig:overhead}. We do the same for a full content-aware IMDN, which requires downloading 4 quality levels $\times$ 3 bitrate levels $=12$ models. We then start to remove models from the enhancement method pool. For example, \emph{low} indicates only keeping low-quality models, and \emph{240p} indicates only keeping enhancement models for 240p. From the results, we find that low-quality content-aware models can provide slightly more QoE than content-agnostic models with a download overhead of 89 KB. And we can gain almost all the benefits by downloading the first two quality levels with only 310 KB. Surprisingly, we can outperform the full model by downloading the first three quality levels using 2215 KB. This implies the ultra-level IMDN model is dragging the performance of \bones by occupying enormous resources but providing marginal improvements.

Assume the downloading processes for models and segments occur sequentially, then downloading models before playback will cause extra startup latency.
To mitigate this, we further introduce a \textit{QuickStart} heuristic that postpones the model downloading and skips enhancement tasks until the buffer level reaches a certain threshold. To avoid rebuffering, we empirically set the threshold to be two segments (8 seconds in our case). We quantitatively analyze the trade-off between performance and overhead in \cref{tab:overhead} using IMDN with its low and medium options on a 3060ti GPU. The content-agnostic enhancement provides the lowest QoE improvement with zero overhead. The content-aware enhancement has the best performance but incurs fine-tuning costs, 310-KB extra download, and 0.68-second additional startup latency. QuickStart sacrifices the possibility of enhancing the first few segments. As a result, it offers an intermediate performance but eliminates the startup latency.

\textit{Other Costs.}
Content-aware enhancement introduces negligible costs of delivering enhancement model addresses in the metadata. Enhancement tasks can slow down the main playback process if they consume excessive computational resources. However, \bones will not favor enhancement under limited computational resources due to the long computation time. Thus, the computational overhead is also not a concern. In summary, \bones significantly improves QoE with minimal overhead, while providing various trade-offs between performance and overhead.

\begin{table}[]
\centering

\caption{Trade-off between performance, offline costs, and online costs. Offline costs include whether training or fine-tuning is required. Online costs include extra download size (KB) and additional startup latency (s).}
\label{tab:overhead}

\begin{tabular}{|c|c|c|c|c|c|c|c|c|c|}
\hline
Method & Qual. & Osc. & Rebuf & QoE & Train & FineTune & Download & Startup \\
\hline
Agnostic & 88.49 & 3.37 & 2.64 & 74.52 & Yes & \textbf{No} & \textbf{0} & \textbf{0} \\
Aware & \textbf{89.02} & \textbf{3.09} & \textbf{2.53} & \textbf{75.80} & Yes & Yes & 310 & 0.68 \\
Aware+QuickStart & 88.80 & 3.14 & 2.56 & 75.40 & Yes & Yes & 310 & \textbf{0}  \\

\hline
\end{tabular}

\end{table}

%% file: prototype.tex
\section{Prototype System}
\label{sec:prototype}

\subsection{System Overview}

\begin{figure}[]
    \centering
    \includegraphics[width=\linewidth]{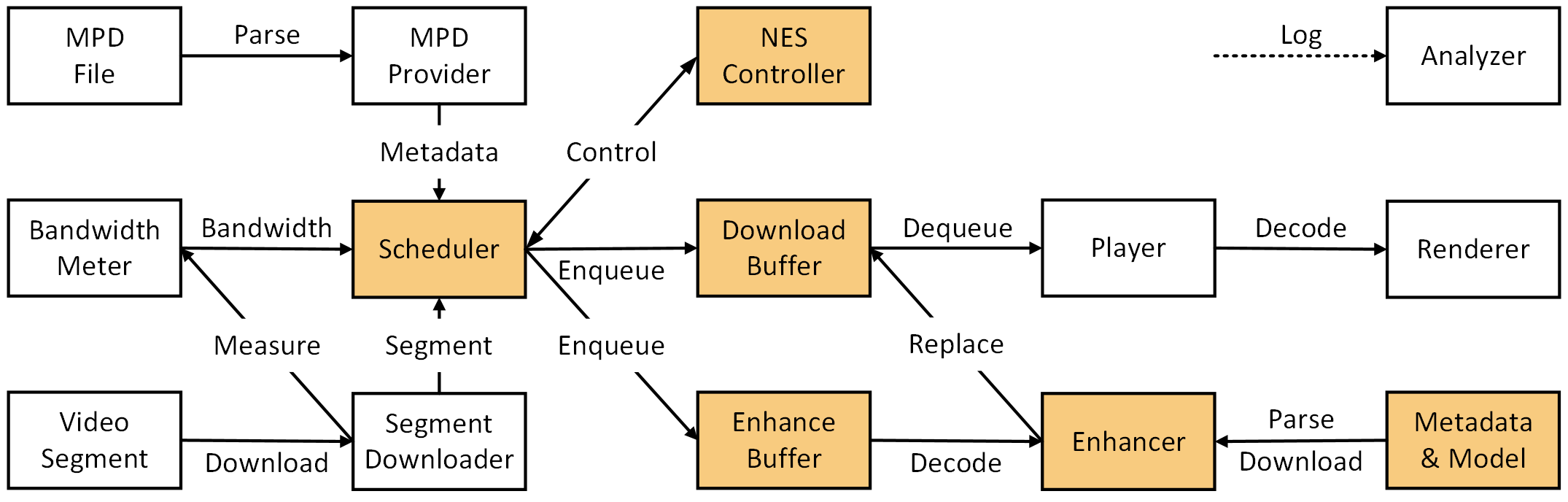}
    \caption{Prototype system architecture. Modules largely different from the ABR pipeline are highlighted.}
    \label{fig:prototype}
\end{figure}

\begin{figure}[]
    \centering
    \includegraphics[width=\linewidth]{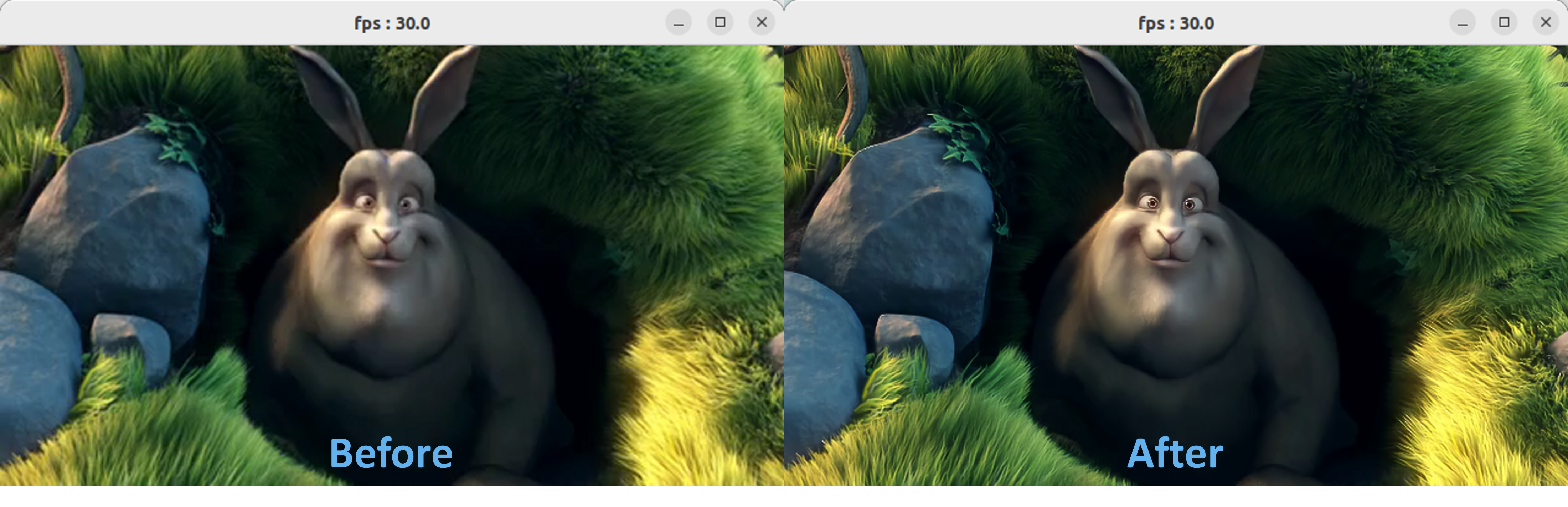}
    \caption{Sample rendering outcome of the prototype system before and after the enhancement enabled by \bones. Compared with the original video frame, the enhanced one features clearer textures and sharper edges, especially for the eyes, grass, and rocks. It is recommended to zoom in for a closer look at the details. The frame rate is displayed in the title portion of the video frame.}
    \label{fig:screenshot}
\end{figure}

To further validate the \bones algorithm in practice, we develop a prototype video player in addition to our simulator. Our prototype implementation is based on the iStream framework~\cite{istream}, which enables the flexible composition of video player modules in a publisher-listener fashion. The prototype system architecture is presented in \cref{fig:prototype}. We highlight the modules significantly different from the traditional ABR pipeline, which we implemented from scratch. Besides, we implemented a graphical renderer for visual comparison, other than the headless player provided by iStream. We implemented the entire system in Python, the deep learning model in PyTorch, the video decoder in VPF \cite{vpf}, and the graphical renderer in OpenGL. Our modification involves approximately 3000 lines of code. In \cref{fig:screenshot}, we illustrate sample frames rendered by the prototype system.

As shown in \cref{fig:prototype}, the system's core component is the \textit{Scheduler}. It first calls the \textit{MPD Provider} to download and parse the metadata. It then interacts with the \textit{NES Controller} to make download and enhancement decisions. Next, it notifies the \textit{Segment Downloader} to download the corresponding segment, and the \textit{Bandwidth Meter} to measure the bandwidth during download. After downloading the video segment, it enqueues the segment into the \textit{Download Buffer} and the \textit{Enhance Buffer}. Meanwhile, the enhancement models will be downloaded if needed. The \textit{Enhancer} fetches the downloaded video segment from the \textit{Enhance Buffer}, decodes, and enhances it. The enhanced segment replaces its counterpart in the \textit{Download Buffer} after enhancement. The \textit{Player} fetches the segment from the \textit{Download Buffer}, decodes it if necessary, and gives it to the \textit{Renderer} for display. The \textit{Analyzer} collects logs from all components all the time.

\subsection{Settings and Experiments}

\begin{table*}[]
\centering

\caption{Performance comparison in the prototype system. \bones achieves the highest QoE.}
\label{tab:prototype}

\begin{tabular}{|c|c|c|c|c|c|c|c|c|c|}
\hline

Method & Qual.  & Osc.  & Rebuf. & QoE & Method & Qual.  & Osc. & Rebuf.  & QoE \\
\hline
\multicolumn{5}{|c|}{Content-Agnostic} & \multicolumn{5}{|c|}{Content-Aware with Quick Start} \\
\hline

\bola*  & 70.44 & \textbf{3.37} & \textbf{0.19} & 66.29 &  \bola* & 70.31 & \textbf{3.47} & \textbf{0.17} & 66.14\\
Dynamic* & 73.23 & 4.22 & 1.60 & 62.57 & Dynamic* & 73.45 & 4.16 & 1.58 & 62.93 \\
Buffer* & 71.24 & 5.24 & 0.27 & 64.89 & Buffer* & 71.35 & 5.24 & 0.17 & 65.41 \\
Throughput* & 69.92 & 3.51 & 1.60 & 59.98 & Throughput*  & 70.15 & 3.68 & 1.58 & 60.10 \\
NAS & 71.65 & 6.92 & 1.45 & 58.92 & NAS & 70.95 & 7.72 & 0.24 & 62.26 \\
\bones & \textbf{75.42} & 3.92 & 0.52 &\textbf{69.38} & \bones & \textbf{76.56} & 3.48 & \textbf{0.17} &\textbf{72.38} \\

\hline
\end{tabular}
\end{table*}

\begin{figure}[]
    \centering
    \includegraphics[width=\linewidth]{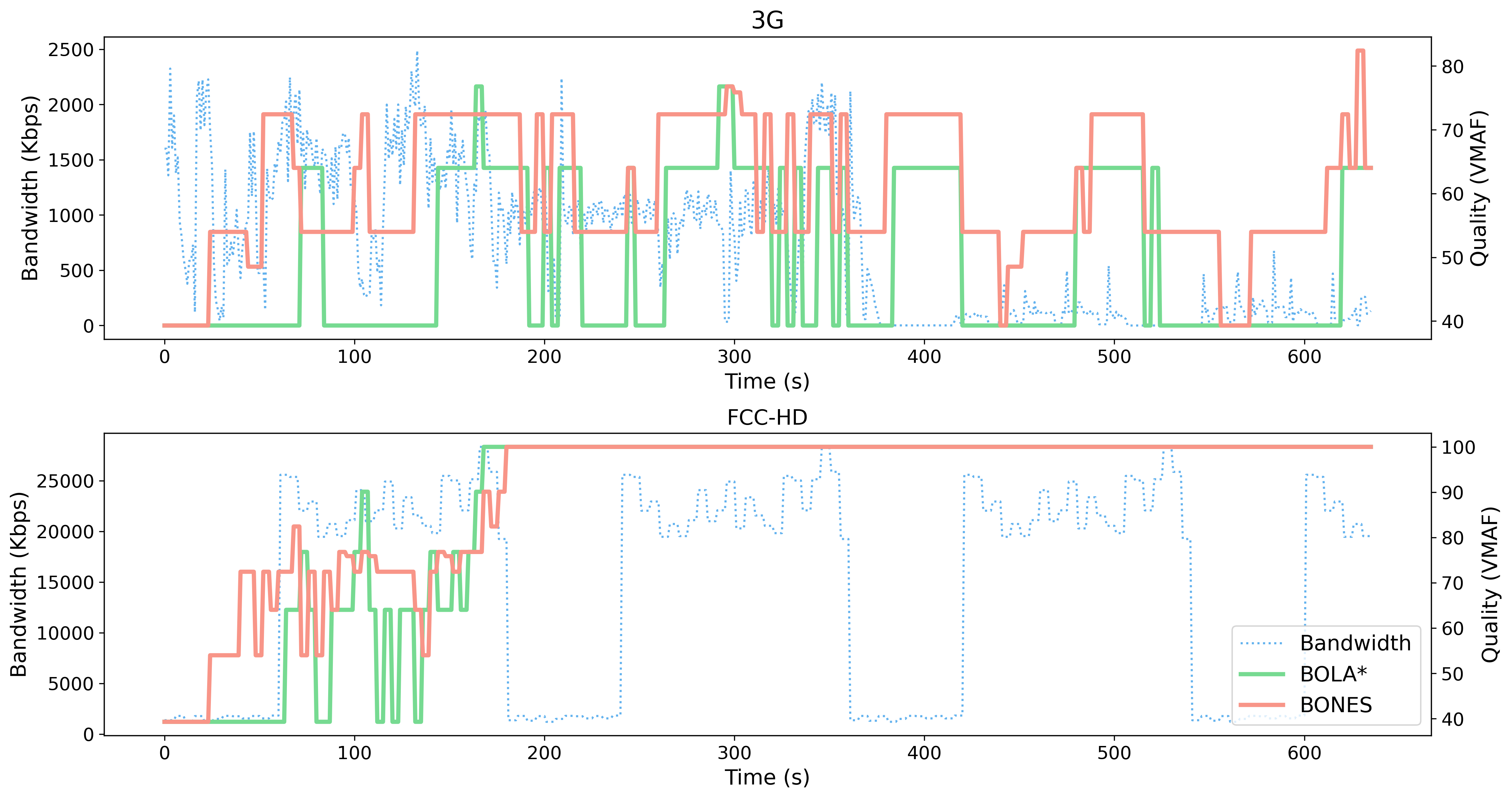}
    \caption{Performance evolution in the prototype system over example traces. }
    \label{fig:istream}
\end{figure}

In our prototype implementation, the enhancement quality ladder is pre-measured and stored locally. The computation speed table is measured the first time this system is deployed. In adaptation to the varying computational resources, we consistently update this table using the latest measurement of enhancement speed. Outdated enhancement tasks will be immediately aborted.

We implement two variants of IMDN enhancer, one content-agnostic and one content-aware with the QuickStart heuristic introduced in \ref{sec:overhead}. We compare our \bones algorithm with the well-performing algorithms in simulation, namely, BOLA*, Dynamic*, Buffer*, Throughput*, and NAS. For a fair comparison, we test the algorithms over 10 pre-recorded network traces sampled from the low-bandwidth 3G network trace dataset and the high-bandwidth FCC-HD dataset. 

We present the overall experimental results in \cref{tab:prototype}. The trend slightly differs from the simulation due to practical factors like the choice of dataset, sampling bias, varying resources, system latency, and implementation details. Under the content-agnostic setting, we observe that \bones outpeforms other baselines in QoE by $4.66\%$ to $13.54\%$. While under the content-aware setting, the advantage of \bones over baselines ranges from $9.43\%$ to $20.43\%$. We also illustrate the performance of \bones and the best baseline BOLA* over two example traces in \cref{fig:istream}, including the bandwidth variation over time and the video quality determined by the algorithms. We find from the figure that \bones generally achieves better quality than the baseline. In conclusion, the advantage of \bones not only holds in large-scale simulations but also translates to practice.

%% file: related.tex
\section{Related Work}

\textbf{Adaptive Bitrate Streaming.}
The ABR algorithm selects the download bitrate for video segments in adaptation to the varying network bandwidth. ABR algorithms can be generally classified into three categories.
First, \textit{buffer-based methods} adapt the bitrate purely based on the client's buffer level. For example, the method in \cite{buffer} chooses the bitrate according to a piecewise linear function of the buffer level. \bola \cite{bola} makes download decisions by solving a Lyapunov optimization problem with respect to the buffer level. Our method could be viewed as an extension of \bola from ABR to the new NES scenario.
Second, \textit{throughput-based methods} utilize the prediction of future network bandwidth. FESTIVE \cite{festive} estimates bandwidth using harmonic mean and chooses the maximum bitrate under that estimation. Fugu \cite{situ} proposes a more accurate download time prediction module based on fully-connected neural networks, while Xatu \cite{situ} employs a long short-term memory (LSTM) network for prediction.
Third, \textit{mixed-input algorithms} consider both buffer level and estimated throughput inputs. Pensieve \cite{pensieve} takes these states as input and trains a reinforcement-learning (RL) agent. FastMPC \cite{mpc} formulates bitrate adaptation as an online control problem and solves it via model predictive control (MPC). Dynamic \cite{bola2} switches between a throughput-based method and \bola to take advantage of both.

\textbf{Neural Enhancement.}
Neural enhancement refers to any method that improves video quality via deep learning. It typically encompasses the following areas.
Super-resolution (SR) aims at improving the resolution of images \cite{mdsr, imdn, swinir} and videos \cite{basicvsr, basicvsr++}. To reduce the training and inference cost of SR models, Li et al. \cite{meta} overfit the first video segment and then fine-tune the model toward the following segments. Wang et al. \cite{revisit} introduce heuristics like training with smaller patches and decreasing the update frequency. DeepStream \cite{deepstream} further introduces scene grouping, frame sub-sampling, and model compression. 
Frame interpolation aims at improving the frame rate of videos \cite{spacetime, unconstrain}. Inpainting methods can complete the missing region of images \cite{lama} or videos \cite{videoinpaint}. Denoising methods effectively reduce noise in videos \cite{videnn, fastdvdnet}. Apart from 2D videos, point cloud upsampling \cite{punet, pugan} increases the density of 3D point cloud objects in volumetric videos, and point cloud completion \cite{pfnet} makes them whole. In this paper, we used NAS-MDSR \cite{nas} and IMDN \cite{imdn} SR models as our enhancement methods, but our control algorithm \bones can be integrated with any subset of the above algorithms.

\textbf{Neural-Enhanced Streaming}
NES incorporates neural enhancement methods into video streaming algorithms and enables high-quality content delivery by leveraging both bandwidth and computational resources. 
NAS \cite{nas} is the first method to integrate SR models into on-demand video streaming. It develops a content-aware SR model called NAS-MDSR and utilizes a Pensieve-like RL agent to manage both download and enhancement processes. SRAVS \cite{improve} adopts RL controllers and a lightweight SR model while proposing a double-buffer system model. PreSR \cite{presr} pre-fetches and enhances ``complex'' segments that bring the most quality improvement and bandwidth reduction. It formulates the problem as MPC and solves it heuristically. Our method is most similar to these approaches.

Neural enhancement can be applied beyond the scope of client-side on-demand video streaming. LiveNAS \cite{livenas} upsamples low-resolution videos at the ingest server for live streaming. NEMO \cite{nemo} increases the computational efficiency of SR by only enhancing selected "anchor" frames. Based on  LiveNAS and NEMO, NeuroScaler \cite{neuroscaler} can enhance live streams at scale. Dejavu \cite{dejavu} improves the quality of the current frame in live video conferencing using historical knowledge. And VISCA \cite{srcache} deploys SR models on edge-based cache servers.
Beyond the regular 2D videos, recent work applied NES to 360-degree and volumetric videos. The substantial bandwidth requirement of these videos makes the assistance of neural enhancement even more appealing. SR360 \cite{sr360} uses RL to make viewport prediction and enhancement decisions for delivering 360 videos. Sophon \cite{sophon} pre-fetches and enhances semantically salient tiles in a 360 video. Madarasingha et al. \cite{edge} enhance 360 videos with both SR and frame interpolation at edge servers. YuZu \cite{yuzu} streams volumetric videos and enhances them using point cloud upsampling. Emerging studies of virtual reality streaming systems also pursue optimal policies for joint allocation of the available computational and communication resources\cite{ChakareskiKRB:21, GuptaCP:20, ChakareskiKY:20}, where the concept of enhancement is not limited to deep-learning methods.

While the above methods are based on RL or heuristic optimization, we formulate and near-optimally solve a Lyapunov optimization problem. Our method \bones is the first NES algorithm with a theoretical performance bound. Furthermore, \bones has been proven to exhibit exceptional performance through numerous experiments, while also possessing the benefits of being simple and robust in deployment.

%% file: conclusion.tex
\section{Conclusion}

This paper proposed \bones, an NES algorithm incorporating neural enhancement into video streaming, allowing users to download low-quality video segments and then enhance them via deep-learning models. Residing in a parallel-buffer system model, \bones jointly optimizes download and enhancement decisions to maximize the QoE by solving a Lyapunov optimization problem. \bones achieves the performance within an additive factor toward offline optimal, making it the first NES algorithm with a theoretical performance guarantee. Experiments verify that \bones can outperform existing ABR and NES methods by $3.56\%$ to $13.20\%$ in simulation and $4.66\%$ to $20.43\%$ in practice. \bones maintains stable performance under all network conditions, has explainable hyper-parameters, and offers different trade-offs between performance and overhead. We publicly release our source code, expecting \bones to become a new baseline for the NES problem.

There are some limitations to the current work, which open future research directions. Firstly, we assume the computation time of enhancement is static or can be accurately predicted. This may not hold in reality especially when other applications are competing for computational resources. Secondly, we only consider scheduling one type of enhancement (super-resolution). In practice, it is possible to simultaneously apply multiple enhancements, which may require the scheduling of multiple enhancement buffers. Thirdly, we only consider a client-side on-demand video streaming scenario. Deploying a streaming algorithm in the cloud or for live streaming may raise other interesting challenges. Given the simplicity of algorithm design, bounded theoretical performance, and superior empirical performance of \bones, in future work, we will explore the possibility of generalizing \bones to more dynamic, multi-buffer, multi-point, and low-latency systems.

%% file: appendix.tex
\appendix

\section{Proof of Theorem 1}
\label{app:proof1}

\begin{proof}

We prove the inequality $Q^d(t_k) \leq V \frac{u_{\max} + \gamma p}{p} + p$ by induction. The bound holds for $k=1$ as $Q^d(t_1) = Q^d(0) = 0$. Suppose it also holds for some $k$. Then there are two cases for $k+1$.

(1) $Q^d(t_k) \leq V \frac{u_{\max} + \gamma p}{p}$. From \cref{eq:Qd} we know $Q^d(t_k)$ can increase by at most $p$ in one single time slot. Hence, we have $Q^d(t_{k+1}) \leq V \frac{u_{\max} + \gamma p}{p} + p$.

(2) $V \frac{u_{\max} + \gamma p}{p} < Q^d(t_k) \leq V \frac{u_{\max} + \gamma p}{p} + p$. In this case, we have 
$$V < \frac{Q^d(t_k) p}{u_{\max} + \gamma p} \leq \frac{Q^d(t_k) p + Q^e(t_k) t^e(i, j)}{u^d(i, t_k) + \Tilde{u}^e(i, j, t_k) + \gamma p}, ~\forall i,j.$$
This makes $O_D(t_k) + V O_P(t_k) > 0, ~\forall \mathbf{d}, \mathbf{e},$ in \cref{eq:alg} and forces the control algorithm not to download. As a result, $Q^d(t_{k+1}) < Q^d(t_k)$. 

Till now, we have proven  $Q^d(t_k) \leq V \frac{u_{\max} + \gamma p}{p} + p$. Combining it with $V \leq \frac{(Q^d_{\max} - p) p}{u_{\max} + \gamma p}$ from assumption, we have $Q_d(t_k) \leq Q^d_{\max}$ and the proof is complete.
\end{proof}

\section{Proof of Theorem 2}
\label{app:proof2}

\begin{lemma}
\label{le:stat}
Denote the optimal objective of problem Eq.~\eqref{eq:opt} as $\Bar{u}^{opt} + \gamma \Bar{s}^{opt}$.
There exists a stationary i.i.d. algorithm for this problem that has the following properties for any $\delta > 0$.
\begin{enumerate}
    \item $\Bar{u}^* + \gamma \Bar{s}^* \leq \Bar{u}^{opt} + \gamma \Bar{s}^{opt} + \delta.$
    \item  $\frac{\mathbb{E} \{ \sum_{ki} d^*_i p \}}{\mathbb{E} \{ \sum_k T^*_k \}} \leq 1 + \delta.$
    \item $\frac{\mathbb{E} \{ \sum_{kij} d^*_i e^*_j t^e(i,j) \}}{\mathbb{E} \{ \sum_k T^*_k \}} \leq 1 + \delta.$
\end{enumerate}
\end{lemma}

\begin{proof}
    Problem Eq.~\eqref{eq:opt} is formulated as maximizing time-averaged objective under time-averaged constraints in a system with variable-size renewal frames. This lemma can be directly derived from Lemma 1 in \cite{renewal}.
\end{proof}

Based on \cref{le:stat}, we prove \cref{the:perf} as follows.

\begin{proof}
\label{prf:bones}

In the following proof, we will use superscript $'$ to denote the decisions of \bones, and $*$ for those of the optimal i.i.d. algorithm.

Following the framework of Lyapunov optimization over renewable frames \cite{renewal}, we define the Lyapunov function $L(t_k)$ as

\begin{equation}
    L(t_k) \triangleq \frac{1}{2} (Q^d(t_k)^2 + Q^e(t_k)^2).
\end{equation}

Define $\mathbf{Q}(t_k) = [Q^d(t_k), Q^e(t_k)]$, then the conditional Lyapunov drift is formulated as

\begin{equation}
    D(t_k) \triangleq \mathbb{E} \{ L(t_{k+1}) - L(t_k) | \mathbf{Q}(t_k)\}.
\end{equation}

The drift $D(t_k)$ is bounded by

\begin{equation}
\label{eq:driftbound} 
\begin{split}
    D(t_k) & \leq  \frac{1}{2} p^2 + \frac{1}{2} (t^e_{\max})^2 + T_{\max} \mathbb{E} \{ T_k |  \mathbf{Q}(t_k) \}  \\
    & + Q^d(t_k) \mathbb{E} \{ \sum_i d_i p - T_k |  \mathbf{Q}(t_k) \} \\
    & + Q^e(t_k) \mathbb{E} \{ \sum_{ij} d_i e_j t^e(i, j) - T_k |  \mathbf{Q}(t_k) \}.
\end{split}
\end{equation}

For the sake of simplicity, we denote $\frac{1}{2} p^2 + \frac{1}{2} (t^e_{\max})^2$ as $C$. Now we reuse $O_D(t_k)$ in \cref{eq:od} and $O_P(t_k)$ in \cref{eq:op}. Adding $V$ times $E\{O_P(t_k)|\mathbf{Q}(t_k)\}$ on both sides and simplifying the formula using $E\{O_D(t_k)|\mathbf{Q}(t_k)\}$, we have

\begin{equation}
\label{eq:orod}
\begin{split}
    & D(t_k) + V \mathbb{E} \{ O_P(t_k) |  \mathbf{Q}(t_k) \} \\
    & \leq C + ( T_{\max} - Q^d(t_k) - Q^e(t_k)) \mathbb{E} \{ T_k |  \mathbf{Q}(t_k) \} \\
    & + \mathbb{E} \{ O_D(t_k) + V O_P(t_k) |  \mathbf{Q}(t_k) \}.
\end{split}
\end{equation}

Recall that the time slot duration can be computed by $T_k = (\sum_i d_i S_i ) / \omega (t_k)$.
As our algorithm greedily minimizes the objective defined in \cref{eq:alg}, we can derive

\begin{equation}
\label{eq:comparestat}
\begin{split}
    & \frac{O_D'(t_k) + V O_P'(t_k)}{ \sum_i d'_i S_i } \leq \frac{O_D^*(t_k) + V O_P^*(t_k)}{ \sum_i d_i^* S_{i} }  \\
    \Rightarrow & \frac{O_D'(t_k) + V O_P'(t_k)}{(\sum_{i} d_{i}' S_{i} ) / \omega (t_k)} \leq
    \frac{ O_D^*(t_k) + V O_P^*(t_k)}{(\sum_{i} d_{i}^* S_{i} ) / \omega (t_k)}  \\
    \Rightarrow & \mathbb{E} \{ O_D'(t_k) + V O_P'(t_k) |  \mathbf{Q}(t_k) \}   \\
    &\leq  \frac{\mathbb{E} \{ T_k' | \mathbf{Q}(t_k) \}}{\mathbb{E} \{ T_k^* | \mathbf{Q}(t_k) \}} \mathbb{E} \{ O_D^*(t_k) + V O_P^*(t_k) |  \mathbf{Q}(t_k) \} . 
\end{split}
\end{equation}

We then apply \cref{eq:comparestat} to \cref{eq:orod} and bring in the definition of \cref{eq:utility}, \cref{eq:smooth}. Since the total video length or total enhancement time is less or equal to the total playback time, we have

\begin{equation}
\label{eq:ineqopt}
\begin{split}
    & D'(t_k) + V \mathbb{E} \{ O_P'(t_k) |  \mathbf{Q}(t_k) \}  \\
    &\leq C + T_{\max} \mathbb{E} \{ T_k' |  \mathbf{Q}(t_k) \} - V (\Bar{u}^* + \gamma \Bar{s}^*) \mathbb{E} \{ T'_k  |  \mathbf{Q}(t_k) \}   \\
     &+ Q^d(t_k) (\mathbb{E} \{ T'_k  |  \mathbf{Q}(t_k) \} (\frac{\mathbb{E} \{ \sum_i d^*_i p |  \mathbf{Q}(t_k) \}}{\mathbb{E} \{ T^*_k  |  \mathbf{Q}(t_k) \}}) - 1 )  \\
    &+ Q^e(t_k) (\mathbb{E} \{ T'_k  |  \mathbf{Q}(t_k) \} (\frac{\mathbb{E} \{ \sum_{ij} d^*_i e^*_j t^e(i,j) |  \mathbf{Q}(t_k) \}}{\mathbb{E} \{ T^*_k  |  \mathbf{Q}(t_k) \}}) - 1 )  \\
    &\leq C + T_{\max} \mathbb{E} \{ T_k' |  \mathbf{Q}(t_k) \} + V (\Bar{u}^* + \gamma \Bar{s}^*) \mathbb{E} \{ T'_k  |  \mathbf{Q}(t_k) \} . 
\end{split}
\end{equation}

Next, sum both sides over $k = 1, \cdots, K_N$ and we get

\begin{equation}
\label{eq:sumover}
\begin{split}
    &  \mathbb{E} \{ L(t_{k+1}) \} + V \mathbb{E} \{ \sum_k O_P'(t_k) \}\\
    & \leq K_N C + T_{\max} \mathbb{E} \{ \sum_k T_k' \} - V (\Bar{u}^* + \gamma \Bar{s}^*) \mathbb{E} \{ \sum_k T'_k  \}  . 
\end{split}
\end{equation}

We divide both sides by $\mathbb{E} \{ \sum_k T'_k  \}$ and take the limit $K_N \to \infty$. Since $L(k_{k+1}) \geq 0$, we have

\begin{equation}
- V (\Bar{u}' + \gamma \Bar{s}') \leq \lim_{K_N \to \infty} \frac{K_N C}{\mathbb{E} \{ \sum_k T_k' \}} + T_{\max} - V (\Bar{u}^* + \gamma \Bar{s}^*) . 
\end{equation}

Because the lower bound of time slot duration is $T_{\min}$, we can transform the above formula into \cref{eq:perf}.

\end{proof}